\documentclass[11pt]{article}
\usepackage{amsmath}
\usepackage{graphicx,psfrag,epsf,epstopdf}
\usepackage{enumerate}

\usepackage{url} 
\usepackage{setspace}
\usepackage[utf8]{inputenc}
\usepackage{verbatim}   
\usepackage{epstopdf}
\usepackage{amsfonts} 
\usepackage{amsthm}
\usepackage{mathtools}
\usepackage{amssymb}
\usepackage{color}
\usepackage{textcomp} 
\usepackage{rotating}
\usepackage{graphicx}
\usepackage{caption}
\usepackage{subcaption}
\usepackage{paralist} 
\usepackage{enumitem}

\usepackage{booktabs}
\usepackage{float}
\usepackage[english]{babel}
\usepackage{verbatim}
\usepackage{epstopdf}
\usepackage{amsfonts} 
\usepackage{amsthm}
\usepackage{amsmath}
\usepackage{amssymb}
\usepackage{color}
\usepackage{textcomp} 
\usepackage{rotating}
\usepackage{graphicx}
\usepackage{caption}
\usepackage{paralist} 
\usepackage[pdftex, pdfborder={0 0 0}]{hyperref} 
\usepackage{booktabs}
\usepackage{float} 
\usepackage{authblk}
\usepackage[font={small,it}]{caption}

\usepackage{epsfig}
\usepackage[round,authoryear,semicolon]{natbib} 

\newtheorem{theorem}{Theorem}[section]

\theoremstyle{definition}

\theoremstyle{plain}
\newtheorem{corollary}[theorem]{Corollary}
\newtheorem{lemma}[theorem]{Lemma}

\newtheorem{assumption}[theorem]{Assumption}

\newtheorem{proposition}[theorem]{Proposition}

\theoremstyle{remark}
\newtheorem{remark}[theorem]{Remark}

\theoremstyle{definition}
\newtheorem{example}[theorem]{Example}

\newcommand\myeq{\mathrel{\overset{\makebox[0pt]{\mbox{\normalfont\tiny\sffamily def}}}{=}}}

\DeclareMathOperator*{\plim}{\mathit{p}-\lim}

\DeclareMathOperator*{\argmin}{argmin}
\DeclareMathOperator*{\argmax}{arg\,max}

\DeclareFontFamily{OT1}{pzc}{}
\DeclareFontShape{OT1}{pzc}{m}{it}{<-> s * [1.10] pzcmi7t}{}
\DeclareMathAlphabet{\mathpzc}{OT1}{pzc}{m}{it}

\newcommand{\overbar}[1]{\mkern 1.5mu\overline{\mkern-1.5mu#1\mkern-1.5mu}\mkern 1.5mu}
\newcommand{\norm}[1]{\left\lVert#1\right\rVert}
\newcommand{\abs}[1]{\left\lvert#1\right\rvert}

\begin{document}

\title{\vspace{-15mm} High-Frequency Volatility Estimation With Fast Multiple Change Points Detection$^*$\thanks{$^*$We are grateful for the helpful comments and suggestions by Gary Kazantsev, W. Brent Lindquist, Stefan Mittnik, Andrew Mullhaupt, Zari Rachev, and Stan Uryasev, as well as the participants of Bloomberg's CTO Data Science Speaker Series, the Financial Math Seminar at Texas Tech University, and the 2023 Winter School in Quantitative Finance at the University of Zurich.} \vspace{-3mm}}
\author[2]{El Mehdi Ainasse $^{\dag}$}
\author[2]{Greeshma Balabhadra \thanks{ These authors contributed equally.}}
\author[1,2,3]{Pawe\l \ Polak \footnote{Corresponding author at Department of Applied Mathematics and Statistics at Stony Brook University, United States. E-mail address: \texttt{pawel.polak@stonybrook.edu}.} \vspace{-2mm}}
\affil[1]{Center of Excellence in Wireless and Information Technology, Stony Brook University, USA \vspace{-2mm}} 
\affil[2]{ Department of Applied Mathematics and Statistics, Stony Brook University, USA \vspace{-2mm}} 
\affil[3]{Institute for Advanced Computational Science, Stony Brook University, USA \vspace{-5mm} }
\date{} 
\maketitle
\vspace{-14mm}
\begin{abstract}
We propose a method for constructing sparse high-frequency volatility estimators that are robust against change points in the spot volatility process. The estimators we propose are $\ell_1$-regularized versions of existing volatility estimators. We focus on power variation estimators as they represent a fundamental class of volatility estimators. We establish consistency of these estimators for the true unobserved volatility and the change points locations, showing that minimax rates can be achieved for particular volatility estimators. The new estimators utilize the computationally efficient least angle regression algorithm for estimation purposes, followed by a reduced dynamic programming step to refine the final number of change points. In terms of numerical performance, these estimators are not only computationally fast but also accurately identify breakpoints near the end of the sample, both features highly desirable in today's electronic trading environment. In terms of out-of-sample volatility prediction, our new estimators provide more realistic and smoother volatility forecasts, outperforming a broad range of classical and recent volatility estimators across various frequencies and forecasting horizons.
\end{abstract}

\section{Introduction}\label{sec:intro}
The detection of sudden shifts in volatility has been extensively studied (\citealp{Sudden-changes-in-variance-and-volatility-persistence-in-foreign-exchange-markets-2003}, \citealp{Sudden-changes-in-volatility-The-case-of-five-central-European-stock-markets-Wang-2009}, \citealp{Detecting-sudden-changes-in-volatility-estimated-from-high-low-and-closing-prices-2013}, and references therein). While some of this work has focused on the detection of structural breaks in volatility for particular classical time-series models such as ARCH and GARCH (\citealp{change-point-detection-garch-2002}, \citealp{changes-structure-garch_2004}, \citealp{structural-breaks-and-garch-models-of-exchange-rate-volatility-2008}, \citealp{ShiHo15}, \citealp{LiNolte:16}, \citealp{semiparametric-approach-detection-change-points-volatility-dynamics-2023}, and references therein), a lot of work on detecting change points in the volatility of (financial) time series has been more general (\citealp{Testing-and-Locating-Variance-Changepoints-with-Application-to-Stock-Prices-1997}, \citealp{Local-polynomial-estimators-of-the-volatility-function-in-nonparametric-autoregression-1997}, \citealp{A-Bayesian-Time-Series-Model-of-Multiple-Structural-Changes-in-Level-Trend-and-Variance-2000}, \citealp{Statistical-inference-for-time-inhomogeneous-volatility-models-2004}, \citealp{Least-squares-estimation-and-tests-of-breaks-in-mean-and-variance-under-misspecification-2004}, \citealp{adaptive-detection-multiple-change-points-asset-price-volatility-Lavielle2007}; and more recently, \citealp{Nonparametric-volatility-change-detection-2021}
, and \citealp{todorov2023testing}). Notably, \citet{Modelling-financial-volatility-in-the-presence-of-abrupt-changes-Ross-2013} showed that modeling financial volatility as a process with structural breaks gives an improved fit compared to commonly used parametric volatility modeling techniques. Detecting changes in volatility has significant applications from a risk perspective \citep{Multiscale-Local-Change-Point-Detection-with-Applications-to-Value-at-Risk-2009} but also in terms of asset allocation \citep{Detecting-change-points-in-VIX-and-SP-500:-A-new-approach-to-dynamic-asset-allocation-Nystrup2016}. 

Under the assumption of no arbitrage, it is known that the price process of an asset is a semi-martingale \citep{DelbaenSchachermayer:94}. Specifically, the dynamics of the $\log$-price process are given by an Itô diffusion in which the volatility process $\sigma$ is typically assumed to be càdlàg. Working in this context, we specifically consider the problem of detecting change points in the spot volatility of an Itô semi-martingale. This problem has been treated in the literature to detect a single change-point (\citealp{Estimation-for-the-change-point-of-volatility-in-a-stochastic-differential-equation-2012}, \citealp{Nonparametric-change-point-analysis-of-volatility-2017}, \citealp{XingLi:22}). We aim to offer a data-driven approach to detect multiple change points without assuming a particular parametric form for the volatility process. 

Recent literature demonstrates that LASSO-based methods yield remarkable performance in terms of volatility forecasting (\citealp{Lassoing-the-HAR-Model-2016}, \citealp{forecasting-realised-volatility-lasso-outperform-har-2021-ding}, \citealp{Quantile-LASSO-in-arbitrage-free-option-markets-MACIAK-2021}, \citealp{MS-MIDAS-LASSO-model-Li2022}, \citealp{ML-Approach-To-Volatility-Forecasting-2022}). Using LASSO for change-point detection (\citealp{HarchaouiLeduc:10}, \citealp{tibshirani-adaptive}), we aim to not only detect change points but also propose spot variance estimators that consistently estimate the true underlying variance process. Although our work can be widely applied to various volatility proxies, much of our exposition (particularly in the empirical section) will focus on the LASSO filtering of volatility proxies based on power variation estimators for integrated variance. This choice is motivated by the ease of exposition (as power variation estimators are heavily used in practice), by the results in \citet{Does-anything-beat-5-minute-RV?-A-comparison-of-realized-measures-across-multiple-asset-classes-2015} who show that 5-minute Realized Variation (RV) is a benchmark that is hard to outperform by any other estimator in a statistically significant way, and by our empirical results which show that the jump robust extension---Bipower Variation (BV) from \citet{NielsenShephard:04}---is the best-performing estimator.

Although our estimator is piecewise constant in the simplest case (i.e. assuming that the volatility itself is of finite total variation), there exist numerous volatility estimators constructed under piecewise stationarity assumptions on the volatility that achieve high forecast accuracy (long-term and short-term) compared to classical parametric estimators -- see \citet{Statistical-inference-for-time-inhomogeneous-volatility-models-2004}, \citet{Nonstationarities-in-Stock-Returns-Starica-2005}, \citet{Recursive-computation-of-piecewise-constant-volatilities-DAVIES-2012}, \citet{On-the-online-estimation-of-local-constant-volatilities-FRIED2012}, \citet{Fryzlewicz:06}, \citet{Modelling-financial-volatility-in-the-presence-of-abrupt-changes-Ross-2013}, \citet{HAR-structural-break-Chin2016-mk}, \citet{MULTIPLE-CHANGE-POINT-DETECTION-FOR-NON-STATIONARY-TIME-SERIES-USING-WILD-BINARY-SEGMENTATION-Fryzlewicz-2017}, \citet{Forecasting-the-term-structure-of-option-implied-volatility:-The-power-of-an-adaptive-method-CHEN-2018}, \citet{heston-piecewise-constant-2018} and references therein. In the same vein as this literature, we also report improvement in volatility forecasting in comparison to parametric estimators (e.g., GARCH-$(1, 1)$, FIGARCH, and realized GARCH), despite piecewise constancy of our LASSO-based estimators. Unlike the existing literature, however, we propose a simpler data-driven approach that does not rely on complex modeling methods. Additionally, we do not propose a particular spot volatility estimator but rather a method for detecting change points in any particular volatility proxy, and we show that the resulting filtered volatility proxy is itself a consistent estimator enjoying superior forecasting properties. As in \citet{HarchaouiLeduc:10}, we recast the multiple change-point estimation problem as a variable selection problem: we use the Least Angle Regression (LARS) algorithm from \citet{EfronHastieJohstoneTibshiraniLAR:04}, and the Dynamic Programming (DP) step to implement the Least Squares Total Variation (LSTV$^*$) algorithm for breakpoint location detection, and simultaneous estimation of the latent spot volatility process. The theory for the proposed estimators is developed, assuming the potential presence of jumps in the price process and microstructure noise. We provide fairly general results, both asymptotic and non-asymptotic, regarding the $\ell_2$-consistency of the proposed LASSO-filtered estimators in terms of fit. We also propose asymptotic results for the consistency of the change points. In particular, we do not necessarily assume that the jumps in the prices or that the microstructure noise is normally distributed. 

The main contributions of our paper consist of developing a new class of estimators for the high-frequency spot volatility process that results from a change-point filtering algorithm. Our algorithm uses discrete samples of a volatility proxy as the observations. There are numerous proxies in the literature, but we essentially categorize them as proxies based on integrated volatility and proxies based on kernels. We show that the LASSO volatility estimator inherits the theoretical guarantees for the consistency of breakpoint location and volatility estimation from the consistency of the original high-frequency volatility proxy without structural breaks. By doing so, we extend the consistency results proposed by \citet{HarchaouiLeduc:10} to the given latent process---the spot volatility of high-frequency returns---and we prove it under more general sub-Weibull distribution assumptions on the $\log$-price jump process and the microstructure noise. Empirically, we show that the forecasting performance of the proposed estimator is much better than a broad spectrum of different volatility models often used in practice for all of the data frequencies, as well as one-step-ahead and daily forecasting horizon periods. In particular, the proposed high-frequency volatility estimator outperforms classical QV and BV by \citet{NielsenShephard:04}; more recent refinements such as realized volatility and bipower variation (RVB) and realized volatility and quarticity (RVQ) by \citet{Yu:20}; the heavy-tailed $t$-GARCH model with Student-$t$ innovations; its long-memory extension ---the $t$-FIGARCH model from \cite{ShiHo15}; as well as the recent TGUH approach to change point detection in high-frequency volatility by \citet{Fryzlewicz:18}; and, finally, the mixed-frequency approach---the HAR model from \cite{Corsi:09}, where the HAR model is used only for multi-step (daily) volatility forecast.

The rest of the paper is organized as follows. In Section \ref{sec:ModelandProblemFormulation}, we introduce various types of volatility proxies under consideration and present the main assumptions and formulation of the model featuring the LASSO-filtered spot volatility process based on a given volatility proxy. Section \ref{sec:Change Point Detection with l1 Regularization} discusses change point detection via $\ell_1$ regularization. Section \ref{sec:TheoreticalResults} outlines our consistency results. In Section \ref{sec:Empirics}, we demonstrate empirical examples using simulated and real data. Section \ref{sec:Extensions} highlights our proposed volatility estimators' versatility and practical application. Finally, in Section \ref{sec:Conclusion}, we conclude. The Online Appendix provides supplementary information, including details on the sub-Weibull distributions, a comparison between the classical LSTV algorithm and our indirect LSTV algorithm, proofs of all the theorems, technical lemmas, and non-asymptotic results for the consistency of the estimated change points.

\section{High-Frequency Volatility Modeling}\label{sec:ModelandProblemFormulation}
Let $p(t)$ denote the stochastic process representing the $log$-price of an asset at time $t\geq 0$. As shown in \citet{DelbaenSchachermayer:94}, under the assumption of no arbitrage,
$p(t)$ must be a semi-martingale ($p\in \mathcal{SM}$),
For all $p\in \mathcal{SM}$, the Quadratic Variation (QV) process can be defined as a probability limit of the sum
\begin{equation}\label{eq:QV}
    \left[p\right](t) = \plim_{n\to \infty} \sum_{i=0}^{n-1}
    \left(p\left(t_{i+1}\right) - p\left(t_{i}\right)\right)^2,
\end{equation}
for any sequence of partitions $t_0 = 0 < t_1 < \ldots < t_n = t$ with $\sup_j (t_{i+1} - t_i)\to 0$ for $n\to \infty$. Note that $p_{t_{i+1}}-p_{t_i}$ is the log-return $r_t$ of the price for $t = t_{i+1}$.

We assume that the $\log$-price process $p$ is a semi-martingale whose local martingale component is given by 
$\int_0^t \sigma(s)dW(s),$
for some function $\sigma$. We also assume that the price has a jump component $J$ with finite activity.

\begin{example}\label{example:merton-jump-diffusion}
Consider the stock price dynamics given by the Merton Jump Diffusion (MJD) process 
\begin{equation}\label{eq:MJD}
\mbox{d} p_t= (\mu(t) - \mu_J \nu)\mbox{d}t + \sigma(t)\mbox{d}W_t + \mbox{d}\sum^{N_t}_{i=1}Q_i
\end{equation}
where $\mu(t)$ and $\sigma(t)$ are drift and spot volatility, respectively, $W_t$ is a standard Brownian motion, and $\sum^{N_t}_{i=1}Q_i$ is a compound Poisson process with normal distributed jumps $Q_i\overset{i.i.d.}{\sim} N(\mu_J,\sigma^{2}_J)$ and intensity $N_t\sim \mathrm{Poisson}(\nu)$, where $\sim$ denotes distributional equivalence. The two processes are independent and adapted to the filtration $\{\mathcal{F}_t\}_{t \geq 0}$, associated with the Brownian motion $\{W_t\}_{t \geq 0}$. Hence, by using It\^{o} isometry for the jump process, one gets 
$$[p](t) = \int_0^t\sigma^2(s)\mbox{d}s + \sum^{N_t}_{i=1}Q_i^2,$$
the sum of Integrated Variance (IV) and the aggregated squared jump components.
\end{example}
More generally, we assume that the observed $\log$-price process at times $t_i := iT/n$ takes the form
$$p(t_i) := p^{*}(t_i)  + \xi(t_i); \quad i = 1, \cdots, n$$
where $p^{*}$ is the true $\log$-price process, $\xi$ represents the micro-structure noise, where $\xi(t_1), \cdots, \xi(t_n)$ are i.i.d. centered sub-Weibull random variables with constant variance $\omega^2$. The true $\log$-price process $p^{*}$ is assumed to satisfy the stochastic differential equation
$$dp(t) := \mu(t)dt + \sigma(t)dW(t) + dJ(t); \quad t \in [0, T],$$
with $\mu(t)$ representing the drift process, $\sigma(t)$ representing the spot volatility, and $J(t)$ is the $\log$-price jump process. Note that typically, $J(t)$ is assumed to be of the form $\sum_{i = 1}^{N(t)} Q_j$, where $N(t)$ is a counting process of finite activity, and the $Q_j$'s are i.i.d. Gaussian random variables. In our case, we replace the Gaussianity assumption by assuming that the $Q_j$'s are sub-Weibull random variables -- an extension of sub-Gaussianity.

\begin{assumption}\label{assumption-cadlag}
Both $\mu$ and $\sigma$ are adapted, càdlàg, and jointly independent of $W$.    
\end{assumption}

\begin{assumption}\label{assumption-boundedness}
    The drift process $\mu$ is predictable and locally bounded, and the volatility process $\sigma$ is locally bounded away from zero.
\end{assumption}

Assumption \ref{assumption-cadlag} consists of regularity conditions on the local behavior of the spot drift and spot volatility functions. This assumption is satisfied by a large number of stochastic volatility models, including those in which $\mu$ and $\sigma$ have continuous trajectories.

Without loss of generality, we can further assume that $\mu$ and $\sigma$ are uniformly bounded and that $\inf_{\tau > 0}\sigma(\tau) > 0$ a.s. (See \citet{andersen2012truncation} and references therein.)

Under the i.i.d. assumptions on the observations of the microstructure noise and the $\log$-price jumps, the $\log$-returns $r_{iT/n} := p(iT/n) - p((i-1)T/n)$ of $p$ satisfy
$$\mathbb{E}\left[r^2_{iT/n}\right] := \int_{(i-1)T/n}^{iT/n} \sigma^2(t) dt + 2\omega^2 + \sum_{(i-1)T/n \leq t < iT/n} (\Delta J(s))^2,$$
where $\Delta J(s) := J(s) - J(s^{-})$.

In the absence of micro-structure noise and jumps, it then follows that
$$\dfrac{n}{T}\mathbb{E}\left[r^2_{iT/n}\right] - \sigma^2((i-1)T/n) \xrightarrow[n \rightarrow +\infty]{} 0,$$
since it follows that
$$\int_{(i-1)T/n}^{iT/n} \sigma^2(t)dt - \dfrac{T}{n}\sigma^2((i-1)T/n) \xrightarrow[n \rightarrow +\infty]{} 0 \quad \mathrm{ a.s.}$$
by Lebesgue's theorem. (See \citealp{andersen2012truncation} for details.)\\

Replacing $r^2_{iT/n}$ by particular transformations of the increments of the log returns allows for jump-robustness and/or robustness against microstructure noise. For instance, particular products of powers of the absolute values of consecutive log-returns are robust against jumps (\citealp{Limit-theorems-for-multipower-variation-in-the-presence-of-jumps-BARNDORFFNIELSEN-2006}, \citealp{multipower-variation}) -- with the classic case being that of consecutive log-returns, encountered in the context of bipower variation as an estimator of the integrated variance (\citet{NielsenShephard:04}). Similarly, truncated squared returns are also jump-robust (\citealp{andersen2012truncation}, \citealp{jacod-reiss-rate}). Further extensions of power variation estimators allow for robustness against both jumps and microstructure noise: \citet{estimation-volatility-functionals-jumps-microstructure-noise-podolskij-2009}. In particular, the presence of jumps in the price process or micro-structure noise does not prevent the construction of consistent spot volatility proxies.\\

For a survey of several estimators of integrated volatility in the presence of microstructure noise and jumps, we refer the reader to \citet{Estimation-Integrated-Volatility-Jumps-Microstructure-Brownless-2020}. These estimators for the integrated volatility are sums of functions of the returns whose increments can be chosen as proxies for the spot volatility. In the next section, we show that the consistency of such integrated volatility estimators (in mean, or a.s.) implies that their increments are consistent proxies of the spot volatilities with respect to Means Squared Error. Note that one can also use intra-day returns to estimate the integrated volatility over a small interval, but these are not under consideration in this paper as we do not assume fixed time intervals but rather assume that our sampling points are equidistant with a distance of $1/n$.

The other type of volatility proxies under consideration are kernel-based spot volatility estimators. These proxies are pointwise estimators of the spot volatility and consistent in the uniform sense (i.e., with respect to the $\sup$-norm). In this case, MSE consistency is immediate. (This will also be briefly discussed in the following section.) Suppose there are no jumps in the $\log$-price process $p$ and no microstructure noise. Consider the truncated realized variance as an estimator of the quadratic variation $[p](t)$ of the $\log$-price in a stochastic volatility model:
$$\widetilde{[p]}_n(t) := \sum_{i = 1}^n \mathbf{1}_{\{t > t_{i-1}\}} r^2_{t_{i-1}} = \sum_{i = 1}^n H(t - t_{i-1}) r^2_{t_{i-1}},$$
where $H$ denotes the Heaviside function. This estimator satisfies
$$\widetilde{[p]}_n(t) \xrightarrow[n \rightarrow +\infty]{\mathbb{P}} [p](t) = \int_0^t \sigma^2(s)ds,$$
and so in particular, for $h > 0$, one has:
$$\dfrac{\widetilde{[p]}_n(t+h) - \widetilde{[p]}_n(t)}{h} \xrightarrow[n \rightarrow +\infty]{\mathbb{P}} \dfrac{1}{h}\int_t^{t+h} \sigma^2(s)ds.$$

Therefore, as $h \rightarrow 0$, the left side of the limit converges to the spot volatility at $t$ plus the square of the jump at $t$. Hence, picking $h:= h_n \xrightarrow[n \rightarrow +\infty]{} 0$, we can approximate the spot volatility (plus squared jump, if any) by
$$\plim_{n \rightarrow +\infty} \dfrac{\widetilde{[p]}_n(t+h_n) - \widetilde{[p]}_n(t)}{h_n}.$$
Note that
$$\dfrac{\widetilde{[p]}_n(t+h) - \widetilde{[p]}_n(t)}{h} = \sum_{i = 1}^n \dfrac{H(t-t_{i-1} + h) - H(t-t_{i-1})}{h} r^2_{t_{i-1}},$$
and the quotient $(H(t-t_{i-1} + h) - H(t-t_{i-1}))/h$
converges to $\delta(t-t_{i-1})$, as $h \rightarrow 0$, where $\delta$ is the Dirac delta function. Therefore, one can extend this idea by using smooth approximations of the Dirac delta function. This leads to a more powerful class of estimators: non-parametric kernel-based spot volatility estimators introduced by \citet{nonparametric-filtering-spot-volatility-kernel}, and more generally spot volatility estimators based on delta sequences introduced by \citet{spot-volatility-delta-sequences} which generalize numerous existing volatility estimators in the literature, and are consistent uniformly in probability, i.e.,
$$\plim_{h \rightarrow 0} \sup_{\tau \in [0, T]} \abs{\tilde{\sigma}^2_h(\tau) - \sigma^2(\tau)} = 0,$$
where $\tau \mapsto \tilde{\sigma}^2_h(\tau)$ refers to the spot volatility estimator in question (which could be kernel-based, built from delta sequences, etc). Strictly speaking, the limit holds for $\tau \in (0, T)$ but can be extended to the boundary points by choosing boundary kernels or local polynomial estimators. The latter is relevant to the context of this paper: Refer to \citet[\S 4]{nonparametric-filtering-spot-volatility-kernel} for more details. In the context of kernel-based spot variance estimators, it is typical to have an additional Hölder-type regularity assumption on the paths of the spot variance function $\tau \mapsto \sigma^2(\tau)$:

\begin{assumption}\label{assumption:holder-paths}
    The spot variance path $\tau \mapsto \sigma^2(\tau)$ lies in the space of functions $\mathcal{C}^{m, \gamma}([0, T])$ for some $m \geq 0$ and $\gamma \in [0, 1]$, where $\mathcal{C}^{m, \gamma}([0, T])$ denotes the space of functions $f$ that are $m$-times differentiable and whose $m$-th derivative satisfies
    $$\abs{f^{(m)}(\tau + \eta) - f^{(m)}(\tau)} \leq L(\tau, \abs{\eta})\abs{\eta}^{\gamma} + o(\abs{\eta}^{\gamma}) \text{ as } \eta \rightarrow 0 \quad a.s.,$$
    where $\eta \mapsto L(\tau, \eta)$ is a slowly varying (random) function at zero, and $\tau \mapsto L(\tau, 0)$ is continuous.
\end{assumption}

Typical construction in the absence of jumps and microstructure noise is the following:
$$\tilde{\sigma}^2(\tau) := \sum_{i = 1}^n K_h(t_{i-1} - \tau) r^2_{t_{i-1}},$$
where $K_h(s) := K(s/h)/h$ for some kernel function $h$. Subject to regularity conditions, for $h := h_n$ converging to $0$ as $n \rightarrow +\infty$ at a suitable rate, $\tilde{\sigma}^2_{h_n}$ is then a consistent estimator of $\sigma^2$ uniformly.

Naturally, a number of regularity assumptions on the kernel need to hold for the spot variance estimators to be consistent:
\begin{assumption}\label{assumption:kernel-regularity}
    The kernel function $K : \mathbb{R} \rightarrow \mathbb{R}$ is continuously differentiable and bounded and satisfies
    \begin{itemize}
        \item $\displaystyle\int_{\mathbb{R}} K(x)dx = 1$;
        \item $\abs{x}\abs{K^{(j)}(x)} \xrightarrow[\abs{x} \rightarrow +\infty]{} 0$ for $j = 0, 1$;
        \item There exists $\Lambda, L < +\infty$ such that $\abs{K^{(i)}(u)} \leq \Lambda$, and for some $\nu > 1$, $\abs{K^{(j)}(u)} \leq \Lambda \abs{u}^{-\nu}$ for $\abs{u} \geq L$, $i = 0, 1$;
        \item $\displaystyle\int_{\mathbb{R}} x^i K(x) dx = 0; \quad i = 0, \cdots, r_K - 1$ and $\displaystyle\int_{\mathbb{R}} \abs{x}^{r_K} \abs{K(x)} dx < +\infty$ for some $r_K \geq 1$.
    \end{itemize}
\end{assumption}

For a survey of such estimators in a more general context -- in the presence of jumps and/or micro-structure -- we refer the reader to \citet{kanaya_kristensen_2016}. This survey, in addition to \citet{nonparametric-filtering-spot-volatility-kernel}, discusses the particular assumptions on the spot drift and volatility processes, showing that certain assumptions (such as independence) are not strictly necessary. We note that numerous existing integrated variance proxies can be adapted as spot volatility proxies, as above, by introducing kernels as weights for the increments of the integrated variance proxy.

Note that this approach for variance estimation has been studied earlier in the more general context of non-parametric regression via the difference sequence method, under similar assumptions to Assumption \ref{assumption:holder-paths} on the paths of $\sigma^2$ -- see \citet{variance-estimation-nonparametric-difference-sequence-2005} and references therein. Functional estimators of the spot variance $\sigma^2$ in the context of Ito diffusions, under assumptions similar to Assumption \ref{assumption:holder-paths}, have been studied even earlier by \citet{adaptive-estimation-diffusion-processes-HOFFMANN-1999}, who also demonstrated their connections to nonparametric regression. Surprisingly, the latter shows that the minimax convergence rate $n^{-q/(1+q)}$ is the same under $L_p$-loss errors for all $p \in [1, +\infty)$ over certain Besov functional spaces $\mathcal{B}_{s, p, \infty}$ where $s > 1 + 1/p$ is the degree of regularity. The result of \citet{nonparametric-filtering-spot-volatility-kernel} extends the result of the $L_{\infty}$-loss, while the minimax result of \citet{variance-estimation-nonparametric-difference-sequence-2005} is a result on the $L_2$-loss.

We further note that a generalization of Assumption \ref{assumption:holder-paths}, \citet[Assumption 3]{optimal-kernel-continuous-FIGUEROALOPEZ2020} (with corresponding adjustments to Assumption \ref{assumption:kernel-regularity}), is satisfied by most volatility processes that are studied in the literature, including fractional Brownian Motion.

The construction of kernel-based spot variance proxies has been further generalized to delta sequences (\citealp{spot-volatility-delta-sequences}), where the kernel $K_{h_n}$ is replaced by a sequence $f_n$ of functions that approximate the Dirac delta. This larger class of estimators includes various functional estimators, including wavelet estimators.

\subsection{On jumps and micro-structure noise}

Microstructure noise can be eliminated by subtracting an estimator $\hat{\omega}^2$ of the variance of the microstructure noise. In the case of proxies based on integrated volatility, the integrated volatility estimator increments can be re-written in a way that includes the subtracted term (see \citealp{Zhang:05}, \citealp{estimation-volatility-functionals-jumps-microstructure-noise-podolskij-2009}). The microstructure noise can also be filtered out via other methods (see \citealp{shephard-designing-realised-kernels-2008}). These approaches can be adapted to the case kernel-based spot variance proxies (see Figure~\ref{fig:AAEASE}). Likewise, the jumps in the $\log$-price process can be filtered out by an appropriate choice of integrated variance increments -- such as consecutive returns (bipower and multi-power variation estimators). The latter can also be adapted to the case of kernel-based spot variance estimators. In what follows, we will assume that the chosen proxy has the desired robustness properties as the proxy remains the user's choice. While there may be methods of comparing various volatility proxies (\citealp{Volatility-forecast-comparison-using-imperfect-volatility-proxies-PATTON2011}), we do not aim to choose the best possible proxy. We would instead aim to filter the volatility levels implied by a given particular proxy using change-point detection via LASSO.

\subsection{A Discussion Of Consistency}\label{subsection:discussion-consistency}
First, let us address spot variance proxies constructed using estimators for integrated variance. Suppose that we are given an estimator of the integrated variance from $0$ to $T$ with increments $\widetilde{\mathrm{IV}}(\tau_i)$, with $i = 1, \cdots, \bar{n}$, where $n$ is the number of samples, and $\bar{n}$ depends on $n$, and $\tau_i := i T/\bar{n}$ for $i = 0, \cdots, \bar{n}$. For instance, for the classical realized variance estimator, $\bar{n} = n$ with $\widetilde{\mathrm{IV}}(\tau_i) := r^2_{iT/n}$. For bi-power variation estimators, $\bar{n} = n$ with $\widetilde{\mathrm{IV}}(\tau_i) := \abs{r_iT/n}\abs{r_{(i+1)T/n}}$. More generally, for multi-power variation estimators, $\bar{n} = n - N + 1$, where $N$ is the number of consecutive returns chosen as increments. In this paper, we only consider $N = 1$ and $N = 2$ for multi-power variation estimators. We note that $\bar{n}$ can also depend on $n$ non-linearly (e.g. $\bar{n} := n^{1/2}$ in \citealp{estimation-volatility-functionals-jumps-microstructure-noise-podolskij-2009}). As per the previous section, we have
$$\sum_{i = 1}^{\bar{n}} \widetilde{\mathrm{IV}}(\tau_i) \xrightarrow[n \rightarrow +\infty]{} \int_0^T \sigma^2(t)dt,$$
where the convergence is assumed to hold in mean or a.s. Note that this assumption, although seemingly strong, does hold for several estimators, including classical estimators (see \citealp[Proof of Theorem 2.1]{Kallenberg1991}; \citealp[\S 2.1 - Theory]{distribution-realized-volatility-ANDERSEN-2001}; and \citealp[Proof of Theorem 3.2]{jacod-reiss-rate}).

Let us denote the integrated variance from $\tau_{i-1}$ to $\tau_i$ by $\mathrm{IV}(\tau_i)$; i.e.
$$\mathrm{IV}(\tau_i) := \int_{\tau_{i-1}}^{\tau_i} \sigma^2(t)dt; \quad i = 1, \dots, n.$$

As remarked before, our local boundedness and regularity assumptions imply that
$$\sqrt{\bar{n}}\sum_{i = 1}^{\bar{n}} \left[\mathrm{IV}(\tau_i) - \dfrac{T}{\bar{n}}\cdot\sigma^2(\tau_{i-1})\right] \xrightarrow[n \rightarrow +\infty]{} 0.$$
(See \citealp{Limit-theorems-for-multipower-variation-in-the-presence-of-jumps-BARNDORFFNIELSEN-2006} or \citealp{andersen2012truncation} for details.)
Therefore, as $\sum_{i = 1}^{\bar{n}} \mathrm{IV}(\tau_i) = \displaystyle\int_0^T \sigma^2(t)dt$, it follows that
$$\sum_{i = 1}^{\bar{n}} \left[\widetilde{\mathrm{IV}}(\tau_i) -\dfrac{T}{\bar{n}}\cdot\sigma^2(\tau_{i-1})\right] \xrightarrow[n \rightarrow +\infty]{} 0.$$

\begin{proposition}\label{proposition:mse-consistency}
    Let $\left\{\tilde{z}_i\right\}_{i = 1}^{\bar{n}}$ and $\left\{z_i\right\}_{i = 1}^{\bar{n}}$ be a sequence of random variables such that
    $\max\left\{O_{\mathbb{P}}\left(\sup_{1 \leq i \leq \bar{n}} \abs{z_i}\right), O_{\mathbb{P}}\left(\sup_{1 \leq i \leq \bar{n}} \abs{\tilde{z}_i}\right)\right\} = O_{\mathbb{P}}(\nu_{\bar{n}})$
    in mean or a.s.
    If $\sum_{i = 1}^{\bar{n}} (\tilde{z}_i - z_i) = o(\bar{n}^{-a})$, in mean or a.s. then $\bar{n}^{-1}\sum_{i = 1}^{\bar{n}} (\tilde{z}_i - z_i)^2 = o_{\mathbb{P}}(\bar{n}^{-a}\nu_{\bar{n}})$.\\ 
    
    Otherwise, if $\sum_{i = 1}^{\bar{n}} (\tilde{z}_i - z_i) = O(\bar{n}^{-a})$, in mean or a.s., then
    $$\dfrac{1}{\bar{n}}\sum_{i = 1}^{\bar{n}} (\tilde{z}_i - z_i)^2 = \begin{cases}
        O_{\mathbb{P}}\left(\bar{n}^{-a}\nu_{\bar{n}}\right), \mathrm{\qquad \qquad  for\, } a < 1, \\
        O_{\mathbb{P}}\left(\bar{n}^{-1}\log(\bar{n})\nu_{\bar{n}}\right), \mathrm{ \quad for\, } a = 1, \\
        O_{\mathbb{P}}\left(\bar{n}^{-1}\nu_{\bar{n}}\right), \mathrm{\qquad \qquad  for\, } a > 1
    \end{cases} \mathrm{ in\, mean\, or\, a.s.}$$
\end{proposition}

In our case, $z_i := \sigma^2(\tau_{i-1})$ and $\tilde{z}_i := \widetilde{\mathrm{IV}}(\tau_i)$.
In particular, the realized variance increments $\widetilde{\mathrm{IV}}(\tau_i)$'s are consistent proxies of the integrated variances $\mathrm{IV}(\tau_i)$'s in MSE, as long as $\nu_{\bar{n}}$ has a suitable growth rate. It is however not clear that the instant variance proxies given by the $\tilde{\sigma}^2_{\mathrm{IV}}(\tau_i) := \bar{n}T^{-1}\widetilde{\mathrm{IV}}(\tau_i)$ would also be consistent estimates of the $\sigma^2(\tau_{i-1})$ spot variances in MSE. Indeed, denote
$\tilde{\boldsymbol{\sigma}}^2_{\mathrm{IV}} := \left(\tilde{\sigma}^2_{\mathrm{IV}}(\tau_0), \cdots, \tilde{\sigma}^2_{\mathrm{IV}}(\tau_{\bar{n}-1})\right)$ and $\boldsymbol{\sigma}^2 := \left(\sigma^2(\tau_0), \cdots, \sigma^2(\tau_{\bar{n}-1})\right),$
and likewise for $\mathbf{IV}$ and $\widetilde{\mathbf{IV}}$. The majority of integrated volatility estimators found in the literature satisfy
$$\sum_{j = 1}^{\bar{n}} \widetilde{\mathrm{IV}}(\tau_i) - \int_0^T \sigma^2(t)dt = O(n^{-a})$$
in mean, for some $a \leq 1/2$ (see e.g. \citealp{jacod-reiss-rate}, \citealp{MANCINI2011845-threshold-speed}, and \citealp{Limit-theorems-for-multipower-variation-in-the-presence-of-jumps-BARNDORFFNIELSEN-2006}). In addition, under assumptions \ref{assumption-cadlag} and \ref{assumption-boundedness}, the supremum of the $\mathrm{IV}(\tau_i)$'s is bounded by $T/\bar{n}$. In the absence of jumps, the supremum of the increments of bipower variation is $O_{\mathbb{P}}(\log(\bar{n})/\bar{n})$. (See \citealp{Limit-theorems-for-multipower-variation-in-the-presence-of-jumps-BARNDORFFNIELSEN-2006}.) 

Alternatively, in the presence of jumps, note that for thresholding estimators, the increments are bounded by the threshold $n^{-\beta}$ for some $\beta \in (0, 1)$ by construction, and such estimators, 
$$\sum_{i = 1}^{\bar{n}} \widetilde{\mathrm{IV}}(\tau_i) - \int_0^T \sigma^2(t)dt = O(n^{-a})$$
holds for $a:= 1-\rho_{\mathrm{BG}}/2$, where $\rho_{\mathrm{BG}} \in (0, 2]$ is the Blumenthal-Getoor index measuring the degree of activity of the small jumps. In that case, $\ell_2$-consistency can be achieved.

Now in the case of kernel estimators with kernel $K$ and bandwidth $h := h_n$, recall that the spot variance kernel proxy $\tau \mapsto \tilde{\sigma}^2_{K, h_n}(\tau)$ satisfies the uniform consistency
$$\sup_{\tau \in [a_n, T - a_n]}\abs{\tilde{\sigma}^2_{K, h_n}(\tau) - \sigma^2(\tau)} \xrightarrow[n \rightarrow +\infty]{\mathbb{P}} 0,$$
under assumptions \ref{assumption:holder-paths}, \ref{assumption:kernel-regularity} and given that $a_n, h_n, a_n/h_n \xrightarrow[n \rightarrow +\infty]{} 0$ at a suitable rate. As we will see later, the highest attainable convergence rate of these spot variance kernel proxies is the minimax rate convergence of adaptive regression splines and adaptive trend filters for suitable bandwidth choices $h_n$. Let $\tilde{\boldsymbol{\sigma}}^2_{K, h_n}$ and $\boldsymbol{\sigma}^2$ denote the vectors
$$\tilde{\boldsymbol{\sigma}}^2_{K, h_n} := \left(\tilde{\sigma}^2_{K, h_n}(t_1), \cdots, \tilde{\sigma}^2_{K, h_n}(t_{n-1})\right)\mathrm{ and\, } \boldsymbol{\sigma}^2 := \left(\sigma^2(t_1), \cdots, \sigma^2(t_{n-1})\right).$$
The uniform consistency above then implies that $\norm{\tilde{\boldsymbol{\sigma}}^2_{K, h_n} - \boldsymbol{\sigma}^2}_{\infty} \xrightarrow[n \rightarrow +\infty]{\mathbb{P}} 0,$
which then implies that $\frac{1}{\sqrt{n}}\norm{\tilde{\boldsymbol{\sigma}}^2_{K, h_n} - \boldsymbol{\sigma}^2}_{2} \xrightarrow[n \rightarrow +\infty]{\mathbb{P}} 0,$
in virtue of the fact that $n^{-1/2}\norm{\cdot}_2 \leq \norm{\cdot}_{\infty}$. The same holds more generally for estimators based on delta sequences. We omit the exact rates, which will be addressed in later sections.

\section{Change Point Detection With \texorpdfstring{$\ell_1$}{l1}-Regularization}\label{sec:Change Point Detection with l1 Regularization}

\subsection{Multiple Change-Point Detection via LASSO}
We are interested in estimating the change points in the latent spot variance process $\sigma^2$ by detecting changes in a consistent proxy $\tilde{\sigma}$ of the latter. As previously discussed, we consider two broad classes of possible proxies:
\begin{itemize}
    \item Those based on integrated variance, whereby a proxy for the integrated variance over short equidistant time intervals is chosen and then appropriately rescaled, resulting in a proxy for the spot variance
    \item Those based on kernels, whereby the spot variance can be directly approximated in a fashion similar to kernel density estimation.
\end{itemize}

In the case of proxies based on integrated volatility, we assume that the true change points locations are $t^{*}_1, \cdots, t^{*}_{K^{*}}$ and that the mean of the $\sigma^2$ is piecewise constant:
$$\mathrm{IV}(t) := \mathrm{IV}^{*}_k; \quad t^{*}_{k-1} + T/\bar{n} \leq t \leq t^{*}_k; \quad k = 1, \cdots, K^{*}; \quad t = T/\bar{n}, \cdots, T,$$
where $t^{*}_0 := \tau_0 = 0$ and $t^{*}_{K^{*}+1} = T$. Likewise, for the proxy $\tilde{\mathrm{IV}}^2$, we assume
$$\widetilde{\mathrm{IV}}(t) := \widetilde{\mathrm{IV}}^{*}_k; \quad \tilde{t}^{*}_{k-1} \leq t \leq \tilde{t}^{*}_k - T/\bar{n}; \quad k = 1, \cdots, \tilde{K}^{*}+1; \quad t = T/\bar{n}, \cdots, T.$$

In the case of kernel-based spot variance proxies, we assume that the true change points locations are $t^{*}_1, \cdots, t^{*}_{K^{*}}$ and that the mean of the $\sigma^2$ is piecewise constant:
$$\sigma^2(t) := s^{2*}_k; \quad t^{*}_{k-1} + T/\bar{n} \leq t \leq t^{*}_k; \quad k = 1, \cdots, K^{*}; \quad t = T/\bar{n}, \cdots, T.$$ 
Likewise, for the proxy $\tilde{\sigma}^2$, we assume
$$\tilde{\sigma}^2_{K, h_n}(t) := \tilde{s}^{2*}_k; \quad \tilde{t}^{*}_{k-1} \leq t \leq \tilde{t}^{*}_k - T/\bar{n}; \quad k = 1, \cdots, \tilde{K}^{*}+1; \quad t = T/\bar{n}, \cdots, T.$$

The change-point locations are estimated by solving the minimization problem
\begin{equation}\label{equation:lasso-filter-optimization}
    \widehat{\boldsymbol{\vartheta}} := \left(\hat{\vartheta}_1, \cdots, \hat{\vartheta}_{\bar{n}}\right) = \min_{\vartheta_1, \cdots, \vartheta_{\bar{n}}} \dfrac{1}{\bar{n}}\sum_{i = 1}^n (\tilde{\vartheta}_i - \vartheta_i)^2 + \lambda_n\sum_{i = 1}^{\bar{n}-1}\abs{\vartheta_{i+1} - \vartheta_i},
\end{equation}

where $\tilde{\vartheta}_i$ denotes either of $\widetilde{\mathrm{IV}}(\tau_i)$ or $\tilde{\sigma}^2_{K, h_n}(\tau_i)$. The change-point locations are then recovered from the jumps in the $\hat{\vartheta}_i$'s. In the integrated volatility case, the estimates of the levels of spot variance are given by $nT^{-1}\widehat{\mathrm{IV}}(\tau_1), \cdots, nT^{-1}\widehat{\mathrm{IV}}(\tau_{\bar{n}})$.\\

This problem is viewed as a LASSO problem since the penalty term can be rewritten as $\norm{\boldsymbol{\beta}}_1$ with $\beta_i:= \vartheta_{i+1} - \vartheta_i$. Let $\mathcal{T}^{*}$ denote the set of change points of $s \mapsto \vartheta(s)$ and let $\widetilde{\mathcal{T}}$ denote the set of change points of $s \mapsto \widetilde{\vartheta}(s)$. We denote the the size of $\mathcal{T}^{*}$ by $\mathfrak{s}^{*}$, and likewise for $\tilde{\mathfrak{s}}$. Note that $\mathfrak{s}^{*}$ and $\tilde{\mathfrak{s}}$ are also the exact number of non-zero components of $\vartheta$ and $\tilde{\vartheta}$ respectively, the latter being known as the zero ``norm'' $\norm{\cdot}_0$. 

As in \citet{Rinaldo_Approximate_change_point_NIPS2017}, let us define the smallest distance between the change points of $\vartheta$ as
$$W_n := \min_{k = 1, \cdots, K^{*}} \abs{t^{*}_{k+1} - t^{*}_k}$$
and the smallest distance between consecutive levels of $\vartheta$ as
$$H_n := \min_{k = 1, \cdots, K^{*}} \abs{\vartheta_{k+1} - \vartheta_k},$$
and define $\tilde{W}_n$ and $\widetilde{H}_n$ similarly.

These quantities are necessary to establish sharp error bounds for $\widehat{\boldsymbol{\vartheta}}$ as an estimator of $\widetilde{\boldsymbol{\vartheta}}$. More specifically, in the noise-less case, the following general proposition on change-point estimators corresponds to \citet[Theorem 1]{Rinaldo_Approximate_change_point_NIPS2017}.

\begin{proposition}\label{proposition:rinaldo-consistency}
    Let $\lambda := \lambda_{\bar{n}} > 0$ be arbitrary. Then $\hat{\boldsymbol{\vartheta}}$ satisfies
    $$\dfrac{1}{\sqrt{\bar{n}}}\norm{\hat{\boldsymbol{\vartheta}} - \tilde{\boldsymbol{\vartheta}}}_2 \leq 8\lambda_n\sqrt{\dfrac{\tilde{\mathfrak{s}}}{\bar{n}\tilde{W}_{\bar{n}}}}.$$
\end{proposition}

This follows from the Basic Inequality in the proof of \citet[Theorem 1]{Rinaldo_Approximate_change_point_NIPS2017} by setting $\epsilon \equiv 0$. The difference between this proposition and \citet[Theorem 1]{Rinaldo_Approximate_change_point_NIPS2017} is that the $\log$-terms do not appear, allowing for potentially fast convergence rates.

In particular, note that for $\widetilde{\mathbf{IV}}$, Proposition \ref{proposition:rinaldo-consistency} implies that
$$\dfrac{1}{\sqrt{\bar{n}}}\norm{\widehat{\mathbf{IV}} - \widetilde{\mathbf{IV}}}_2 \leq 8\lambda_n\sqrt{\dfrac{\tilde{\mathfrak{s}}}{\bar{n}\tilde{W}_{\bar{n}}}},$$
and so, in particular, it follows that
$$\dfrac{1}{\sqrt{\bar{n}}}\norm{\widehat{\boldsymbol{\sigma}}^2_{\mathrm{IV}} - \bar{n}T^{-1}\widetilde{\mathbf{IV}}}_2 \leq 8T^{-1}\lambda_n\sqrt{\dfrac{\bar{n}\tilde{\mathfrak{s}}}{\tilde{W}_{\bar{n}}}}.$$

In the case of integrated volatility proxies, contamination by microstructure noise and jumps can be introduced as follows.

\begin{assumption}\label{assumption:noise-contamination-integrated-variance}
Suppose that we are given an estimator of integrated volatility that is not robust against jumps and microstructure noise in the following way:
$$\sum_{i = 1}^{\bar{n}} \widetilde{\mathrm{IV}}(\tau_i) \xrightarrow[n \rightarrow +\infty]{} \int_0^T (\sigma^2(s)+\omega^2(s))ds + \sum_{i = 1}^{N_T} Q^2_j,$$
where $\omega^2$ denotes the variance of the micro-structure noise, which we assume to be bounded uniformly: $\max_{s \in [0, T]} \omega^2(s) < +\infty$, $N_T$ denotes a stationary point process, and the $Q_j$'s are i.i.d. random variables with finite fourth moment. (See \citealp{estimation-volatility-functionals-jumps-microstructure-noise-podolskij-2009} for instance.) Under this assumption, we consider this model for the true latent integrated volatility:
$$\mathrm{IV}(t) := \mathrm{IV}^{*} + \xi(t); \quad \tilde{t}^{*}_{k-1} \leq t \leq \tilde{t}^{*}_k - T/\bar{n}; \quad k = 1, \cdots, \tilde{K}^{*}+1; \quad t = T/\bar{n}, \cdots, T,$$
where $\tilde{t}^{*}_0 := \tau_0 = 0$, $\tilde{t}^{*}_{\tilde{K}^{*}+1} = T$, and
$\xi(\tau_i) := \displaystyle\int_{\tau_{i-1}}^{\tau_i} \omega^2(s)ds + \sum_{j = 1 + N_{\tau_{i-1}}}^{N_{\tau_i}} Q^2_j; \quad j = 1, \cdots, \bar{n}.$

Note that although the $\xi(\tau_i)$'s are not centered, they are i.i.d. with distribution
$$\displaystyle\int_{\tau_{i-1}}^{\tau_i} \omega^2(s)ds + \sum_{j = 1 + N_{\tau_{i-1}}}^{N_{\tau_i}} Q^2_j \sim \displaystyle\int_{\tau_{i-1}}^{\tau_i} \omega^2(s)ds + \sum_{j = 1}^{N_{\tau_i} - N_{\tau_{i-1}}} Q^2_j \sim \displaystyle\int_{\tau_{i-1}}^{\tau_i} \omega^2(s)ds + \sum_{j = 1}^{N_{T/\bar{n}}} Q^2_j.$$

Denoting by $Q$ the common distribution of the $Q_j$'s and using Wald's identities, we have
$$\mathbb{E}\left[\sum_{j = 1}^{N_{T/\bar{n}}} Q^2_j\right] = O(T/\bar{n}), \text{ and } \mathbb{E}\left[\left(\sum_{j = 1}^{N_{T/\bar{n}}} Q^2_j\right)^2\right] = O(T/\bar{n})(\mathbb{E}[Q^4] + \mathbb{E}[Q^2]^2) \text{ as } \bar{n} \rightarrow +\infty$$
by \citep[Theorem 3.3.1, Theorem 3.4.1]{beutler-stationary-point-processes-1966}. (See also \citealp{random-sampling-stationary-point-processes-BEUTLER} for a concrete exposition with numerous examples, including Poisson renewal processes.)

Since $\omega^2$ is assumed to be uniformly bounded, it follows that the first and second moments of each $\xi(\tau_i)$ are, in fact, $O(\bar{n}^{-1})$.
\end{assumption}

\subsection{The LARS-LSTV Algorithm}
To estimate the piecewise constant signal \citet{HarchaouiLeduc:10}, adapt the Least Angle Regression (LARS) algorithm to solve the LSTV objective function from \eqref{equation:lasso-filter-optimization}. The $\ell_1$-regularized estimators are obtained using the algorithm below. 
 
In particular, our change points detection method consists of two steps. First is the following LSTV algorithm, which takes $K=K_{\mathrm{max}}$, an upper bound for the number of breakpoints, and returns the most significant $K_{\mathrm{max}}$ change points candidates.
\begin{enumerate}
    \item \textit{Initialization}: $k=0$; set $\widehat{\mathcal{T}}_{\bar{n},0}=\emptyset$ (change points locations) and $\hat{\vartheta}_{i}^{[0]} = 0$, for $i=1,\ldots,\bar{n}$ (estimated integrated/spot variance).
    \item While $k< K_{\mathrm{max}}$:
    \begin{itemize}
    \item \textit{Change point addition}: 
    Find $\hat{t}_k$  such that
    $$\hat{t}_k = \argmax_{\tau\in\left\{0,\frac{T}{\bar{n}},\ldots,\frac{(\bar{n}-1)T}{\bar{n}}\right\}\setminus\widehat{\mathcal{T}}_{\bar{n},k-1}} \left| \sum_{i=\tau}^{\bar{n}}\vartheta_i-\sum_{i=\tau}^{\bar{n}} \hat{\vartheta}_i^{[k-1]}\right|.$$
    Update the set of change points candidates: $\widehat{\mathcal{T}}_{\bar{n},k} = \widehat{\mathcal{T}}_{\bar{n},k-1} \cup\left\{\hat{t}_k\right\}$. 
     \item \textit{Update}: $\hat{\vartheta}_i^{[k]}$, for $i=1,\ldots,\bar{n}$, as a piecewise constant fit to the $\theta_i$'s with change points at $\widehat{\mathcal{T}}_{\bar{n},k}$, i.e., for all sorted $\hat{t}^{(j)}\in \widehat{\mathcal{T}}_{\bar{n},k}$, where $j=1\ldots,k$,
     $$ \hat{\vartheta}_{m}^{[k]} \myeq(\hat{t}^{(j)}-\hat{t}^{(j-1)})^{-1} \sum_{i=\hat{t}^{(j)}+1}^{\hat{t}^{(j)}} \vartheta_i\quad \forall m=\hat{t}^{(j-1)}+1,\ldots,\hat{t}^{(j)}.$$
    \item \textit{Descent direction computation}: 
    Compute $\mathbf{w}_k = (\mathbf{X}_k^T\mathbf{X}_k)^{-1} \mathbf{1}_k$, where $\mathbf{X}_k$ is a matrix which consists of the columns of $\mathbf{X}$ indexed by the elements of $\widehat{\mathcal{T}}_{\bar{n},k}$.
    \item \textit{Descent step search}: Search for $\widehat{\gamma}$ such that
    $$\widehat{\gamma} = \min_{\tau\in\left\{0,\frac{t}{\bar{n}},\ldots,\frac{(\bar{n}-1)T}{\bar{n}}\right\}\setminus \widehat{\mathcal{T}}_{\bar{n},k}}\left(\frac{\sum_{i=\tau}^{\bar{n}}\vartheta_i- \sum_{i=\tau}^{\bar{n}}\hat{\vartheta}_i^{[k]})}{1 - \sum_{i=\tau}^{\bar{n}} w_{k,i}},\frac{\sum_{i=\tau}^{\bar{n}}\vartheta_i+ \sum_{i=\tau}^{\bar{n}}\hat{\vartheta}_i^{[k]}}{1 + \sum_{i=\tau}^{\bar{n}} w_{k,i}}\right).$$
    \item \textit{Zero-crossing check}: Let $\alpha_j = \mbox{sign}(\hat{\vartheta}_{j+1}^{[k]} - \hat{\vartheta}_{j}^{[k]})$, if $$\widehat{\gamma}>\widetilde{\gamma} \myeq \min_{j} (\alpha_j w_{k,j})^{-1} \left(\sum_{i=j}^{\bar{n}} \hat{\vartheta}_i^{[k]}\right),$$
    then, decrease $\widehat{\gamma}$ down to $\widetilde{\gamma}$, and remove $\widetilde{\tau}$ from $\widehat{\mathcal{T}}_{\bar{n},k}$, where 
    $$\widetilde{\tau} \myeq \arg \min_{j\in \widehat{\mathcal{T}}_{\bar{n},k}}(\alpha_j w_{k,j})^{-1} \left(\sum_{i=j}^{\bar{n}} \hat{\vartheta}_i^{[k]}\right).$$
   \end{itemize}
\end{enumerate}

 The second step consists of the reduced Dynamic Programming (rDP) method to determine the final and optimal set of change points. First, for every $K\leq K_{\mathrm{max}}$, the rDP step finds $K$ change points by searching through the set $\widehat{\mathcal{T}}_{n,K_{\mathrm{max}}}=\left\{t_1,\ldots,t_{K_{\mathrm{max}}}\right\}$ instead of all of the points $\{\tau_1,\ldots, \tau_{\bar{n}}\}$. The objective function for the rDP step for each $K$ in $\{1,\ldots,K_{\mathrm{max}}\}$ is 
$$ J(K) = \min_{t_1<\ldots<t_{K} }\sum_{k=1}^{K} \sum_{i=t_{k-1}+1}^{t_k} (\vartheta_i-\hat{\vartheta}_k)^2, \text{ where } \hat{\vartheta}_k \myeq (t_k-t_{k-1})^{-1} \sum_{i=t_{k-1}+1}^{t_k} \vartheta_i.$$
Hence, for every $K\leq K_{\mathrm{max}}$, the selected $K$ change points reduce the sum of square errors the most among all the $K$ combinations of the candidate change points. For the model selection, i.e., to select the optimal number of change points, following \citet{HarchaouiLeduc:10}, we use $\rho_k = J(k+1)/J(k)$ and pick $\widehat{K} : \widehat{K} = \min_{k \geq 1} \{ \rho_k \geq 1-\xi \} $, where $\xi$ is the model selection threshold parameter, and in all of the empirical analysis below it is set to $\xi=0.03$. The two-step algorithm we denote by LSTV$^*$(QV) and LSTV$^*$(BV) for quadratic variation and bi-power variation increments, respectively, and the rDP step on the LSTV results.

Empirically, we observe that the rDP step, with $O(K_{\mathrm{max}}^3)$ complexity, eliminates false change point candidates from the LSTV step, which has $O(K_{\mathrm{max}}n\log(n))$ complexity. Since $K_{\mathrm{max}}$ is the maximum number of plausible breakpoints, which is much smaller than $n$, the algorithm has an overall complexity of $O(K_{\mathrm{max}}n\log(n))$. 

\subsection{Connections to Adaptive Trend Filtering}
We now highlight connections between the LASSO-based LSTV multiple change-point detection and adaptive trend filtering (\citealp{tibshirani-adaptive}). The latter establishes minimax convergence properties of adaptive piecewise polynomial estimators, which we will discuss momentarily.

Let $\left\{y_i\right\}_{i = 1}^{\bar{n}}$ be a time series modeled as
$$y_j := f(\tau_i) + \varepsilon_j; \quad j = 1, \cdots, \bar{n},$$
where the times $\tau_1, \cdots, \tau_{\bar{n}}$ are equidistant, i.e. $\tau_{i+1} - \tau_i := T/\bar{n}$, and $\varepsilon_1, \cdots, \varepsilon_{\bar{n}}$ are centered i.i.d. sub-Gaussian random variables of order $\alpha \leq 2$. Now consider the following sets for $\kappa \geq 0$:
\begin{itemize}
    \item The knot super set $\mathcal{T}_{\kappa} := \begin{cases}
        \left\{\tau_{\kappa/2+2}, \cdots, \tau_{\bar{n}-\kappa/2}\right\} \quad &\mathrm{ if } \kappa \text{ is even},\\
        \left\{\tau_{(\kappa+1)/2}, \cdots, \tau_{\bar{n}-(\kappa+1)/2}\right\} \quad &\mathrm{ if } \kappa \text{ is odd}.
    \end{cases}$;
    \item The set of restricted locally adaptive regression splines 
    $$\mathcal{G}_{\kappa} := \left\{g : [0, T] \rightarrow \mathbb{R} : g \text{ is a } \kappa\text{-th degree spline with knots contained in } \mathcal{T}_{\kappa}\right\};$$
    \item The set of unrestricted locally adaptive regression splines 
    $$\mathcal{F}_{\kappa} := \left\{g : [0, T] \rightarrow \mathbb{R} : g \text{ is } \kappa\text{ times weakly differentiable and } \mathrm{TV}\left(g^{(\kappa)}\right) < +\infty\right\}.$$
\end{itemize}
where $\kappa$-th degree splines are defined as piece-wise polynomial functions of order at most $\kappa$, and $\mathrm{TV}(h)$ denotes the total variation of a function $h : [0, T] \rightarrow \mathbb{R}$ defined by
$$\text{TV}(h) := \sup\left\{\sum_{i = 1}^{P} \abs{h(\tau_{i+1}) - h(\tau_i)} : \tau_1 < \cdots < \tau_P \mathrm{\, is\, a\, partition\, of\, } [0, T]\right\} < +\infty.$$

The $\kappa$-th order locally adaptive regression estimate is given by
$$\hat{f}_{\mathcal{G}} = \min_{f \in \mathcal{G}_{\kappa}} \dfrac{1}{\bar{n}}\sum_{i = 1}^{\bar{n}} (y_i - f(\tau_i))^2 + \dfrac{\lambda}{\kappa!}\cdot \mathrm{TV}\left(f^{(\kappa)}\right),$$
while its unrestricted counterpart is given by
$$\hat{f}_{\mathcal{F}} \in \min_{f \in \mathcal{F}_{\kappa}} \dfrac{1}{\bar{n}}\sum_{i = 1}^{\bar{n}} (y_i - f(\tau_i))^2 + \dfrac{\lambda}{\kappa!}\cdot \mathrm{TV}\left(f^{(\kappa)}\right).$$
Note that the set membership symbol for the unrestricted case emphasizes that the solution might not be unique. It is, however, unique for $\kappa = 0$ and $\kappa = 1$. The corresponding discretization for both problems is
\begin{equation}\label{equation:discrete-trend-filter}
    (\hat{\vartheta}_1, \cdots, \hat{\vartheta}_{\bar{n}}) = \min_{\boldsymbol{\vartheta} = (\vartheta_1, \cdots, \vartheta_{\bar{n}}) \in \mathbb{R}^{\bar{n}}} \dfrac{1}{\bar{n}}\sum_{i = 1}^{\bar{n}} (y_i - \vartheta_i)^2 + \dfrac{\lambda}{\kappa!}\norm{D^{(\kappa+1)}\boldsymbol{\vartheta}}_1,
\end{equation}
where $D^{(\kappa + 1)} \in \mathbb{R}^{(\bar{n}-\kappa-1)\times \bar{n}}$ is the discrete difference operator of order $\kappa$. When $\kappa = 0$,
$$D^{(1)} = \begin{pmatrix}
    -1 & 1 & 0 & \cdots & 0 & 0\\
    0 & -1 & 1 & \cdots & 0 & 0\\
    \vdots & \vdots & \vdots & \cdots & \vdots & \vdots\\
    0 & 0 & 0 & \cdots & -1 & 1
\end{pmatrix} \in \mathbb{R}^{(\bar{n}-1) \times \bar{n}},$$
and $D^{(\kappa+1)} = D^{(1)} \cdot D^{(\kappa)}$ for $\kappa \geq 1$. Equivalently,
$$D^{(1)}\boldsymbol{\vartheta} = (\vartheta_2 - \vartheta_1, \cdots, \vartheta_{\bar{n}} - \vartheta_{\bar{n}-1}) \text{ and } D^{(\kappa+1)}\boldsymbol{\vartheta} = D^{(1)}\left(D^{(k)}\boldsymbol{\vartheta}\right) \text{ for } \kappa \geq 1.$$

The restricted locally adaptive regression spline estimate $\hat{f}_{\mathcal{G}}$ and its unrestricted counterpart $\hat{f}_{\mathcal{F}}$ coincide for $\kappa = 0$ and $\kappa = 1$. In addition, the discrete estimates $\hat{\vartheta}_1, \cdots, \hat{\vartheta}_{\bar{n}}$ coincide with the locally adaptive regression spline estimates at the sampling times, meaning
$$\hat{\vartheta}_i := \hat{f}(\tau_i); \quad i = 1, \cdots, \bar{n},$$
with $\hat{f}$ denoting either of the locally adaptive regression estimates. For $\kappa \geq 2$, the estimates are generically different but asymptotically equivalent in MSE. Moreover, assuming that the true function $f$ belong to the class $\mathcal{F}_{\kappa}$, the estimates given by \ref{equation:discrete-trend-filter} converge to the values of the true function $f$ in MSE at the minimax rate $O_{\mathbb{P}}(n^{-(2\kappa+2)/(2\kappa+3)})$ for $\lambda = \Theta(n^{1/(2\kappa+3)})$.

As it should be clear by this point, the LSTV multiple change-point detection is identical to the discrete $\kappa=0$ order trend filtering. The $\kappa=1$ order trend filtering is also of natural consideration in our context in that one may seek to detect change points in the spot variance by detecting trends in the integrated variance.

In our context, the time series $\left\{y_i\right\}_{i = 1}^{\bar{n}}$ represents the series of samples of the proxies for the integrated/spot variance, but without the sub-Gaussian noise contamination as one may choose proxies that are robust against both jumps and microstructure noise. However, we allow for noise contamination in the case of integrated variance under Assumption \ref{assumption:noise-contamination-integrated-variance}.

We refer the reader to \citet{tibshirani-adaptive} for further details and a detailed exposition.

\section{Consistency Results}\label{sec:TheoreticalResults}

In this section, we discuss the theoretical guarantees of our proposed LASSO estimators. We cover minimax results and concentration results.

\subsection{Consistency Results for Estimators Based on Integrated Variance}
We start a general proposition and then describe a number of general results regarding estimators based on integrated variance and spot variance estimators derived from those. We defer results regarding kernel-based estimators to the next as those are systematically minimax.

We first discuss the case of estimates based on integrated volatility. We focus on the case where the proxy is robust against jumps and microstructure noise.

As before, let $\mathrm{IV}(\tau_i) := \displaystyle\int_{\tau_{i-1}}^{\tau_i} \sigma^2(s)ds$. For convenience, we consider the root MSE rather than MSE to analyze the consistency of the change-point detection estimates. Let $\left(\widehat{\mathrm{IV}}(\tau_1), \cdots, \widehat{\mathrm{IV}}(\tau_{\bar{n}})\right)$ denote the solution to Problem \ref{equation:discrete-trend-filter} with $\widetilde{\mathrm{IV}}(\tau_1), \cdots, \widetilde{\mathrm{IV}}(\tau_{\bar{n}})$ as the target values. Now let $\mathbf{IV}$ denote the vector of integrated variances and similarly for $\widetilde{\mathbf{IV}}$ and $\widehat{\mathbf{IV}}$. Furthermore, let $\boldsymbol{\sigma}^2$ denote the vector of spot variances. By the triangle inequality, we have
$$\dfrac{1}{\sqrt{\bar{n}}}\norm{\widehat{\mathbf{IV}} - \dfrac{T}{\bar{n}}\boldsymbol{\sigma}^2}_2 \leq \dfrac{1}{\sqrt{\bar{n}}}\norm{\widehat{\mathbf{IV}} - \widetilde{\mathbf{IV}}}_2  + \dfrac{1}{\sqrt{\bar{n}}}\norm{\widetilde{\mathbf{IV}} - \mathbf{IV}}_2 + \dfrac{1}{\sqrt{\bar{n}}}\norm{\mathbf{IV} - \dfrac{T}{\bar{n}}\boldsymbol{\sigma}^2}_2.$$

\begin{proposition}\label{proposition:consistency-integrated}
    Let $\widetilde{\mathrm{IV}}(\tau_1), \cdots, \widetilde{\mathrm{IV}}(\tau_{\bar{n}})$ be increments of an integrated volatility proxy satisfying
    $\bar{n}^{-1}\norm{\widehat{\mathbf{IV}} - \widetilde{\mathbf{IV}}}^2_2 = O_{\mathbb{P}}(\bar{n}^{-\tilde{a}})$. Suppose that $\sigma^2$ satisfies 
    $\bar{n}^{-1}\norm{\mathbf{IV} - \dfrac{T}{\bar{n}}\boldsymbol{\sigma}^2}^2_2 = O_{\mathbb{P}}(\bar{n}^{-a})$.
    If $\bar{n}^{-1}\norm{\widehat{\mathbf{IV}} - \widetilde{\mathbf{IV}}}^2_2 = O_{\mathbb{P}}(\bar{n}^{-\hat{a}})$, for a given choice of $\lambda_n$, then $\widehat{\mathbf{IV}}$ is a consistent estimator of $T\bar{n}^{-1}\boldsymbol{\sigma}^2$ in MSE with a rate of $O_{\mathbb{P}}(\bar{n}^{-\omega})$ where $\omega = \min\left\{\tilde{a}_1 + \tilde{a}_2, a_1 + a_2, \hat{a}\right\}$. Moreover, for the same choices of $\alpha$, $\beta$ and $\hat{a}$, $\hat{\boldsymbol{\sigma}}^2_{\mathrm{IV}} := \bar{n}T^{-1} \cdot \widehat{\mathbf{IV}}$ is a consistent estimator of $\boldsymbol{\sigma}^2$ in MSE with a rate of $O_{\mathbb{P}}(\bar{n}^{-\omega+1})$ for $\omega > 1$.
\end{proposition}

\begin{corollary}\label{corollary:consistency-integrated-minimax-boundedness}
Let $\tilde{a}_1, \tilde{a}_2$ and $\hat{a}$ be as in Proposition \ref{proposition:consistency-integrated}. If $\sigma^2$ satisfies assumptions \ref{assumption-boundedness} and \ref{assumption-cadlag}, then $a_1 = 2$ and $a_2 = 1$, and hence Proposition \ref{proposition:consistency-integrated} holds with $\omega := \min\left\{\tilde{a}_1 + \tilde{a}_2, 3, \hat{a}\right\}$. In particular, if $\min\left\{\tilde{a}_1 + \tilde{a}_2, \hat{a}\right\} = 1/(2\kappa+3)$, then $\hat{\boldsymbol{\sigma}}^2_{\mathrm{IV}} := \bar{n}T^{-1} \cdot \widehat{\mathbf{IV}}$ is a minimax estimator of $\boldsymbol{\sigma}^2$ in the class of locally adaptive regression splines of order $\kappa$.
\end{corollary}

Typically, we expect $\alpha \leq 1/2$ for most integrated volatility estimators, and $\beta < 2$ in the absence of jumps, while Proposition \ref{proposition:rinaldo-consistency} shows that we can expect $\widehat{\mathbf{IV}}$ to be a consistent estimator of $\mathbf{IV}$ for $\lambda_{\bar{n}} := \bar{n}^{-\hat{a}}(\bar{n}\tilde{W}_{\bar{n}})^{-1/2}\tilde{\mathfrak{s}}^{1/2}$. In general, we can let $\hat{a} = \tilde{a}_1 + \tilde{a}_2$ in $\lambda_{\bar{n}}$.

\begin{example}
    For a suitably chosen $\hat{a}$, Proposition \ref{proposition:consistency-integrated} and Corollary \ref{corollary:consistency-integrated-minimax-boundedness} hold for multi-power variation estimators (\citealp{Limit-theorems-for-multipower-variation-in-the-presence-of-jumps-BARNDORFFNIELSEN-2006}) whose powers add up to $2$ and whose maximum $M$ satisfies $M < 2$, where $\tilde{a}_1$ and $\tilde{a}_2$ are given by $\tilde{a}_1 = 1-M/2 = \tilde{a}_2$. This holds in the presence of jumps but does not assume the presence of micro-structure noise. In particular, if $M < 2 - (2\kappa + 3)^{-1}$, then $\hat{\boldsymbol{\sigma}}^2_{\mathrm{IV}} := \bar{n}T^{-1} \cdot \widehat{\mathbf{IV}}$ is a minimax estimator of $\boldsymbol{\sigma}^2$ in the class of locally adaptive regression splines of order $\kappa$.
\end{example}

\begin{example}
    Proposition \ref{proposition:consistency-integrated} and Corollary \ref{corollary:consistency-integrated-minimax-boundedness} hold for thresholding estimators (\citealp{MANCINI2011845-threshold-speed}) with threshold $\bar{n}^{-\beta}$, $\beta \in (0, 1)$. For such estimators, the rate of convergences is $\bar{n}^{-\beta(1-\rho_{\mathrm{BG}}/2)}$ for $\rho_{\mathrm{BG}} < 1$ and $\beta \leq (2-\rho_{\mathrm{BG}})^{-1}$, or for $\rho_{\mathrm{BG}} \geq 1$, where $\rho_{\mathrm{BG}} \in (0, 2]$ is the Blumenthal-Getoor index. This holds in the presence of jumps but does not assume the presence of micro-structure noise.
\end{example}

Note that for the previous two examples, $\bar{n}$ is of the order of $n$, and so the rates with respect to $n$ hold with respect to $\bar{n}$. This may, of course, not be the case for certain proxies (e.g., \citealp{estimation-volatility-functionals-jumps-microstructure-noise-podolskij-2009} where $\bar{n} = n^{1/2}$).

\begin{example}
    The realized kernels of \citet{shephard-designing-realised-kernels-2008} also satisfy Proposition \ref{proposition:consistency-integrated} and Corollary \ref{corollary:consistency-integrated-minimax-boundedness} for numerous choices of weighing functions (kernels) and bandwidths. This assumes the presence of both jumps and microstructure noise.
\end{example}

Now consider the non-robust case in the form of Assumption \ref{assumption:noise-contamination-integrated-variance}. Let $\boldsymbol{\xi}$ denote the vector whose components are given by
$$\xi(\tau_i) = \int_{\tau_{i-1}}^{\tau_i} \omega^2(s)ds + \sum_{j = 1 + N_{\tau_{i-1}}}^{N_{\tau_i}} Q^2_j.$$

Recall that the first and second and moments of the $\xi(\tau_i)$'s are $O\left(\bar{n}^{-1}\right)$.
In that case, $\mathbf{IV}$ is replaced by $\mathbf{IV} + \boldsymbol{\xi}$, and we have:
$$\dfrac{1}{\sqrt{\bar{n}}}\norm{\widehat{\mathbf{IV}} - \dfrac{T}{\bar{n}}\boldsymbol{\sigma}^2}_2 \leq \dfrac{1}{\sqrt{\bar{n}}}\norm{\widehat{\mathbf{IV}} - \widetilde{\mathbf{IV}}}_2 + \dfrac{1}{\sqrt{\bar{n}}}\norm{\widetilde{\mathbf{IV}} - (\mathbf{IV}+\boldsymbol{\xi})}_2 + \dfrac{1}{\sqrt{\bar{n}}}\norm{\boldsymbol{\xi}}_2 + \dfrac{1}{\sqrt{\bar{n}}}\norm{\mathbf{IV} - \dfrac{T}{\bar{n}}\boldsymbol{\sigma}^2}_2.$$

Compared to the contamination-free case, all we need to do is ensure that $\bar{n}^{-1/2}\norm{\boldsymbol{\xi}}_2$ goes to zero sufficiently fast. The latter is a matter of guaranteeing suitable convergence rates for sums of i.i.d. random variables. More specifically, we have:

$$\bar{n}^{-1}\norm{\boldsymbol{\xi}}^2_2 = \bar{n}^{-1}\sum_{i = 1}^{\bar{n}} \xi^2(\tau_i) = \bar{n}^{-1}\sum_{i = 1}^{\bar{n}} \overbar{\xi^2(\tau_i)} + O\left(\bar{n}^{-1}\right),$$
where $\overbar{Z} := Z-\mathbb{E}[Z]$.

\begin{proposition}

Let $\widetilde{\mathrm{IV}}(\tau_1), \cdots, \widetilde{\mathrm{IV}}(\tau_{\bar{n}})$ be increments of an integrated volatility proxy satisfying $\bar{n}^{-1}\norm{\tilde{\mathbf{IV}}-(\mathbf{IV} + \boldsymbol{\xi})}^2_2 = O_{\mathbb{P}}(\bar{n}^{-\tilde{a}})$. Suppose further that $\sigma^2$ satisfies $\bar{n}^{-1}\norm{\mathbf{IV}-T\bar{n}^{-1}\boldsymbol{\sigma}^2}^2_2 = O_{\mathbb{P}}(n^{-a}),$
    and that $\bar{n}^{-1}\sum_{i = 1}^{\bar{n}} \overbar{\xi^2(\tau_i)} = O_{\mathbb{P}}(\bar{n}^{-a_{\xi}})$.
    If $\bar{n}^{-1}\norm{\widehat{\mathbf{IV}} - \widetilde{\mathbf{IV}}}^2_2 = O_{\mathbb{P}}(\bar{n}^{-\hat{a}})$, for a given choice of $\lambda_n$, then $\widehat{\mathbf{IV}}$ is a consistent estimator of $T\bar{n}^{-1}\boldsymbol{\sigma}^2$ in MSE with a rate of $O_{\mathbb{P}}(\bar{n}^{-\varpi})$ where $\varpi = \min\left\{\tilde{a}, a, a_{\xi}, 1, \hat{a}\right\}$.
\end{proposition}

The growth rate of $\bar{n}^{-1}\sum_{i = 1}^{\bar{n}} \overbar{\xi^2(\tau_i)}$ can be determined using general results on sums of random variables such as \cite{convergence-rates-law-large-numbers-baum-katz-1965}, \citet{klass-teicher-iterated-law-asymmetric-1977}, as well as \citet{convergence-sum-negatively-associated-random-variables-Liang2010} and references therein. These results allow us to establish convergence rates under assumptions on the moments of the summands.

\begin{remark}
    Note that since $\varpi \leq 1$, the consistency of $\hat{\boldsymbol{\sigma}}^2_{\mathrm{IV}} := \bar{n}T^{-1} \cdot \widehat{\mathbf{IV}}$ as an estimator of $\boldsymbol{\sigma}^2$ cannot be achieved with respect to MSE. 
\end{remark}

\subsection{Consistency Results for Kernel-Based Estimators}
In this section, we consider results for kernel-based estimators. Let $K$ be a kernel satisfying Assumption \ref{assumption:kernel-regularity} and let $\tilde{\boldsymbol{\sigma}}^2_{K, h_n}$ and $\boldsymbol{\sigma}^2$ denote the vectors
$$\tilde{\boldsymbol{\sigma}}^2_{K, h_n} := \left(\tilde{\sigma}^2_{K, h_n}(t_1), \cdots, \tilde{\sigma}^2_{K, h_n}(t_{n-1})\right)\mathrm{ and\, } \boldsymbol{\sigma}^2 := \left(\sigma^2(t_1), \cdots, \sigma^2(t_{n-1})\right),$$
and let $\hat{\boldsymbol{\sigma}}^2_{K, h_n}$ be the adaptive trend filtering estimate of $\tilde{\boldsymbol{\sigma}}^2_{K, h_n}$. Note that the estimator $\tilde{\sigma}^2_{K, h_n}$ is not assumed to be jump-robust or noise-free.

Then by \citet[Theorem 3]{nonparametric-filtering-spot-volatility-kernel}, we have:

\begin{proposition}\label{proposition:mse-consistency-kernel}
    Suppose that the $\sigma^2$ satisfies Assumption \ref{assumption:holder-paths} for some $m \in \mathbb{N}_{\geq 0}$ and $\gamma \in [0, 1]$ and let $K$ be a kernel satisfying Assumption \ref{assumption:kernel-regularity} with $r_K \geq m + \gamma$. If $h_n$ be a kernel bandwidth such that $(nh_n)^{-1} \xrightarrow[n \rightarrow +\infty]{} 0$, then $\tilde{\boldsymbol{\sigma}}^2_{K, h_n}$ is a consistent estimator of $\boldsymbol{\sigma}^2$ in MSE with a rate of
    $$O_{\mathbb{P}}\left(h^{m + \gamma}_n\right) + O_{\mathbb{P}}\left(\dfrac{\log(n)}{\sqrt{nh_n}}\right) + O_{\mathbb{P}}\left(\hat{R}_n\right),$$
    where $\hat{R}_n$ is the rate of convergence of $\hat{\boldsymbol{\sigma}}^2_{K, h_n}$ to $\tilde{\boldsymbol{\sigma}}^2_{K, h_n}$ for a given $\lambda = \lambda_n$.
\end{proposition}

\begin{corollary}\label{corollary:mse-consistency-kernel}
    Suppose that the assumptions of Proposition \ref{proposition:mse-consistency-kernel} on $\sigma^2$ and $K$ hold, and let $h_n = O(n^{-1/(2(m+\gamma)+1)})$ and $\lambda_n := \Theta(n^{1/(2\kappa+3)})$. Then for $m = \kappa$ and $\gamma = 1$, $\hat{\boldsymbol{\sigma}}^2_{K, h_n}$ is a minimax estimator of $\boldsymbol{\sigma}^2$ in MSE in the class of locally adaptive regression splines of order $\kappa$.
\end{corollary}

Note that the convergence rates in Proposition \ref{proposition:mse-consistency-kernel} and Corollary \ref{corollary:mse-consistency-kernel} depend on $n$, not $\bar{n}$ (which could be a power of $n$), unlike our results on estimators based on integrated variance. Moreover, the kernel-based estimators satisfy more general functional assumptions whose regularity can be directly tied to the order of the spline estimator.

These minimax rates of convergence are also achieved by a similar class of functional estimators studied by \citet{adaptive-estimation-diffusion-processes-HOFFMANN-1999}, as well as those based on the difference sequence method (\citealp{variance-estimation-nonparametric-difference-sequence-2005}) and those based on the delta sequences (\citealp{spot-volatility-delta-sequences}.

Note that these results can be extended to similarly constructed estimators that are robust to micro-structure and/or jumps: see \citet{kanaya_kristensen_2016}, \citet{spot-volatility-delta-sequences}, \citet{kernel-estimation-jump-Cai2020}, \citet{figueroa-lópez_wu_2022}.

\subsection{Consistent Estimation of Change Points Locations}\label{change points-locations-section}

We present results on the consistent estimation of the change point locations. Let $\mathcal{T}$ denote the set of change points in the latent volatility process $\sigma^2$, i.e.
$$\mathcal{T} := \left\{i \in \{1, \cdots, m\} : \sigma^2(\tau_i) \neq \sigma^2(\tau_{i+1})\right\},$$
where $m = \bar{n}$ or $m = n$, depending on the type of volatility proxy. The sets $\widetilde{\mathcal{T}}$ and $\widehat{\mathcal{T}}$ are defined similarly.

For two discrete sets $A$ and $B$, define the metrics
\begin{equation}\label{equation:hausdorff distance}
    d(A|B) := \max_{b\in B}\min_{a \in A} \abs{a-b} \text{ and } d_H(A, B) := \max\left\{d(A|B), d(B|A)\right\}.
\end{equation}

The first metric is a one-sided screening distance from $B$ to $A$, measuring the furthest distance of an element in $B$ to its closest element in $A$. It can be seen as a measure of how well the set $B$ covers the set $A$. The second metric, the Hausdorff distance, is a symmetric distance between $A$ and $B$.

We now recall two theorems of \citet{Rinaldo_Approximate_change_point_NIPS2017} regarding approximate change point screening. These theorems are generic, assuming no specific data model or any particular assumptions on the estimator. The first theorem is the following.

 \begin{theorem}\citep[Theorem 4]{Rinaldo_Approximate_change_point_NIPS2017}\label{theorem:rinaldo-changepoint-screening}
     Let $\tilde{\boldsymbol\vartheta} = (\tilde{\vartheta}(\tau_1), \cdots, \tilde{\vartheta}(\tau_m))$ be an estimator for $\boldsymbol\vartheta = (\vartheta(\tau_1), \cdots, \vartheta(\tau_m))$ such that $m^{-1}\norm{\tilde{\boldsymbol\vartheta} - \boldsymbol\vartheta}^2_2 = O_{P}(\tilde{R}_m)$. Let $W_m := \min_{k \in \{1, \cdots, K^{*}\}} \abs{t^{*}_{k+1} - t^{*}_k}$ be the smallest distance between the change points of $\vartheta$ and let $H_m := \min_{k \in \{1, \cdots, K^{*}\}} \abs{\vartheta_{k+1}-\vartheta_k}$, where $K^{*}$ is the true number of change points of $\vartheta$. If $m\tilde{R}_m H^{-2}_m = o(W_m)$, then
     $$d\left(\widetilde{\mathcal{T}} \big\vert \mathcal{T}\right) = O_{\mathbb{P}}(m\tilde{R}_m H^{-2}_m).$$
 \end{theorem}

\citet{Rinaldo_Approximate_change_point_NIPS2017} propose a post-screening procedure for the estimated change points that eliminate estimated change points that lie far away from true change points. This procedure is based on a further filtering of $\tilde{\vartheta}$:
$$F_i(\tilde{\vartheta}) := b^{-1}_m \cdot \left(\sum_{j = i + 1}^{i + b_m} \tilde{\vartheta}_j - \sum_{j = i - b_m + 1}^i \tilde{\vartheta}_j\right) \text{ for } i = b_m, \cdots, m - b_m,$$
for an integer bandwidth $b_m > 0$.

\begin{theorem}\citep[Theorem 5]{Rinaldo_Approximate_change_point_NIPS2017}\label{theorem:rinaldo-changepoint-recovery-hausdorff}
    Let $\tilde{\vartheta}$ be such that $m^{-1}\norm{\widetilde{\boldsymbol\vartheta}-\boldsymbol\vartheta}^2_2 = O_{\mathbb{P}}(\tilde{R}_m)$ and define the evaluation locations
    $$\mathcal{I}_F := \left\{i \in \{b_m, \cdots, m - b_m\} : i \in \mathcal{T} \text{ or } i + b_m \in \mathcal{T} \text{ or } i - b_m \in \mathcal{T} \right\} \cup \{b_m, m - b_m\}.$$
    Let the set $\widetilde{\mathcal{S}}_F(\widetilde{T}_m) := \left\{ i \in \mathcal{I}_F : \abs{F_i(\tilde{\vartheta})} \geq  \widetilde{T}_m   \right\} $   for a threshold level $\widetilde{T}_m$. If $b_m, \widetilde{T}_m$ satisfy
    $$mR_m H^{-2}_m b^{-1}_m \xrightarrow[m \rightarrow +\infty]{} 0, \quad 2b_m \leq W_m, \quad \text{ and } \lim_{m \rightarrow +\infty} \dfrac{\widetilde{T}_m}{H_m} \in (0, 1),$$
    then the filtered change points satisfy
    $$\mathbb{P}\left(d_H\left(\mathcal{S}_F(\widetilde{T}_m), \mathcal{T}\right) \leq 2 b_m\right) \xrightarrow[m \rightarrow +\infty]{} 1,$$
    where, $W_m$ and $H_m$ are as defined in Theorem \ref{theorem:rinaldo-changepoint-screening}.
\end{theorem}

 Note that the minimum separations between change points for $\bar{n}^{-1}T\sigma^2$ and $\sigma^2$ are identical, and we denote both by $W_{\bar{n}}$. Denoting by $H_{\bar{n}}$ the minimum separation between the levels of $\sigma^2$, it is clear that $T\bar{n}^{-1}H_{\bar{n}}$ is the minimum separation between the levels of $\sigma^2$.

\begin{proposition}\label{proposition:change-point-location-consistency-integrated-variance}
    Let $\widetilde{\mathbf{IV}}$ denote an estimator of the integrated variance $\mathbf{IV}$. Suppose that the assumptions of Proposition \ref{proposition:mse-consistency} hold, and let $O_{\mathbb{P}}\left(\bar{n}^{-\omega}\right)$ be the rate of convergence of $\widehat{\mathbf{IV}}$ to $T\bar{n}^{-1}\boldsymbol\sigma^2$ in MSE. If $\bar{n}^{-\omega-1}H^{-2}_{\bar{n}} = o(W_{\bar{n}})$, then
    the set of estimated change points $\widehat{\mathcal{T}}$ satisfies
    $$d\left(\widehat{\mathcal{T}} \big\vert \mathcal{T}\right) = O_{\mathbb{P}}(\bar{n}^{-\omega-1}H^{-2}_{\bar{n}}).$$
\end{proposition}

\begin{proposition}\label{proposition:change-point-location-consistency-kernel}
    Let $\widetilde{\boldsymbol\sigma}^2_{K, h}$ denote an estimator of the spot variance $\boldsymbol\sigma^2$. Suppose the assumptions of Proposition \ref{proposition:mse-consistency-kernel} hold. If $n^{-\omega-1}H^{-2}_{n} = o(W_{n})$, then as $(nh_n)^{-1} \xrightarrow[n \rightarrow +\infty]{} 0$,
    the set of estimated change points $\widehat{\mathcal{T}}$ satisfies
    $$d\left(\widehat{\mathcal{T}} \big\vert \mathcal{T}\right) = O_{\mathbb{P}}\left(\left(h^{m+\gamma}_n + \dfrac{\log(n)}{\sqrt{nh_n}} + \hat{R}_n\right)H^{-2}_n\right),$$
    where $\hat{R}_n$ is the rate of convergence of $\widehat{\boldsymbol\sigma}^2_{K, h}$ to $\widetilde{\boldsymbol\sigma}^2_{K, h}$ for a given $\lambda = \lambda_n$. In particular, for $\lambda_n = 
 \Theta(n^{1/(2(\kappa+1)+1})$ and $h_n = O(n^{-1/(2(m+\gamma)+1)})$ with $m = \kappa$ and $\gamma = 1$, we have
 $$d\left(\widehat{\mathcal{T}} \big\vert \mathcal{T}\right) = O_{\mathbb{P}}\left(n^{-1/(2(\kappa+1)+1}\right).$$
\end{proposition}

\section{Empirical Analysis} \label{sec:Empirics}
In the first simulation, we confirm that our estimators can accurately capture the spot volatility process, even when multiple change points are present. We calculate the $\log$-returns $r_{1,t}, \ldots, r_{n,t}$ from the simulated price data, employing the MJD model from \eqref{eq:MJD} with $S_0=1$, $\mu=0.02$, $\nu = 1$, $\mu_J = 0.0$, $\sigma_J=0.015$. These parameters are taken from \citet{ait2013leverage} and correspond to a one-minute sampling frequency. We simulate $n=3900$ observations corresponding to approximately ten trading days of 1-minute data. Spot volatility is modeled as a step function $\sigma(s)$, featuring five break points and six levels. Namely,  
$\sigma(s) \in \{2.12, 1.51, 2.35, 1.83, 2.44, 1.65, 3.13\} \times 10^{-4}.$ These values are calibrated from NYSE TAQ Apple minute data from January 2016 to February 2022, each listed above as the one-year median of minute-level standard deviation. We display the simulated price series alongside the corresponding $\log$-returns in Figure~\ref{fig:stockprices}---the change points in the volatility are not easily discernible from the price or returns series in the figure. The bottom panel in Figure~\ref{fig:stockprices} displays the corresponding squared returns together with the comparison between the actual spot variance and the estimated values using the LSTV$^*$(BV) change points detection method with $\xi=0.3$, and $K_{\mathrm{max}}=8$, i.e., $K_{\mathrm{max}}> K^*=5$. The estimator successfully detects all the breakpoints and provides a close fit with a small mean-square error between the true and estimated spot volatility. As for LSTV$^*$(QV), it yields a very similar fit and is excluded from the plot.

\begin{figure}[ht]
\begin{center}
\includegraphics[width=0.9\textwidth, height=0.45\textwidth]{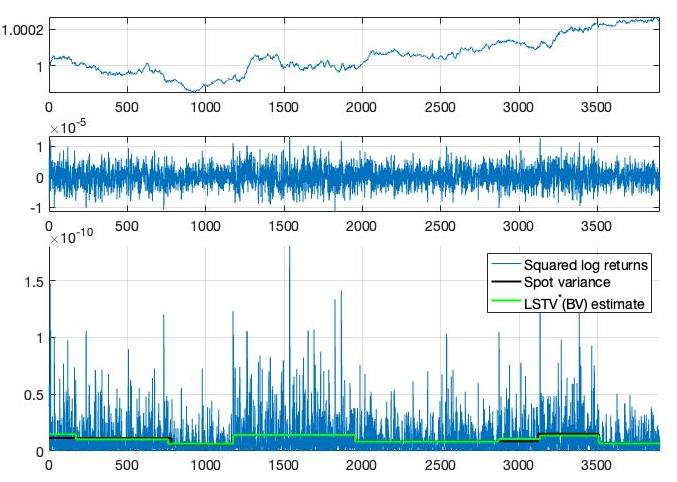}
\caption{Simulated stock prices using the MJD dynamic given in Example \ref{example:merton-jump-diffusion} with six levels of the spot volatility  $\sigma = \{2.12, 1.51, 2.35, 1.83, 2.44, 1.65, 3.13\}*10^{-4}$,  $\mu=0.02$, $\nu = 1$, $\mu_J = 0.0$, $\sigma_J=0.015$. \textbf{Top figure:} Simulated stock price values. \textbf{Middle figure:} Corresponding $\log$-returns, where the volatility regimes are visible with breakpoints at $780$, $1170$, $1950$, $3120$, and $3510$. \textbf{Bottom figure:} Log-returns squared, and a comparison of LSTV$^{*}$(BV) estimator fit vs.\ the simulated spot variance process. \normalsize}
\label{fig:stockprices}
\end{center}
\end{figure}

\begin{figure}[ht]
\begin{center}
\includegraphics[width=0.95\textwidth, height=0.7\textwidth]{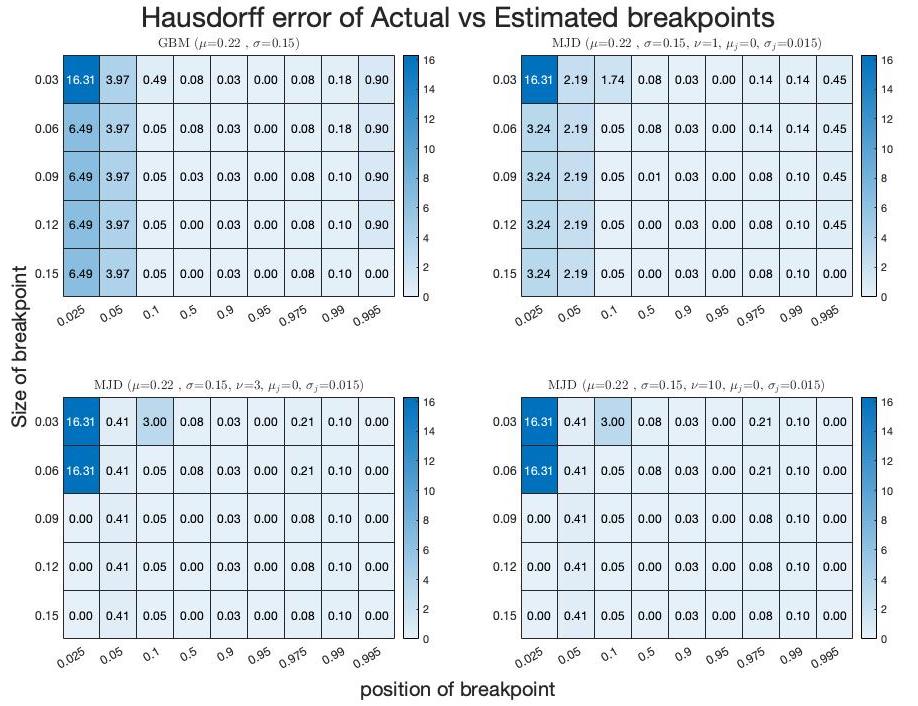}
\caption{ Hausdorff error comparison of LSTV$^*$(BV) estimator fit for simulated stock prices with jumps and fixed levels of volatility with the position of jump size vs. changes in the size of jumps for different MJD models from Example \ref{example:merton-jump-diffusion}. \textbf{Top left:} heat map is for classical Geometric Brownian Motion (MJD with no jumps) simulated stock prices of total sample size $3900$ with $\mu=0.22$ and various levels of sigma changes per row - $(\sigma_1, \sigma_2)$ are $(0.15,0.18)$, $(0.15,0.21)$, $(0.15,0.24)$, $(0.15,0.27)$, $(0.15,0.3)$, respectively. Each column indicates the position of the breakpoint in percentiles of actual data - $0.01$, $0.025$, $0.1$, $0.5$, $0.95$, $0.995$, $0.999$. \textbf{Top right, bottom left}, and \textbf{bottom right:} heat maps are MJD simulations with $\mu_J = 0.0, \sigma_J=0.015 $ and jump intensities $\nu=1,3,10$, respectively.}
\label{fig:Hausdorff}

\end{center}
\end{figure}

Next, to evaluate the sensitivity of the LSTV$^*$ method to the size and location of the change point in the data, we simulate a single breakpoint at various positions and sizes within the simulated data. Our simulation focuses on locating the breakpoint rather than estimating the correct number of breakpoints. Therefore, we set $K_{\mathrm{max}}=1$ and keep $\xi=0.3$ in the LSTV$^*$ algorithm. We calculate the average Hausdorff distance between true and estimated change points using  \autoref{equation:hausdorff distance}
over $10000$ random paths from the MJD model, using the same parameters as in the first simulation. To assess the robustness of our results, we set the jump intensity $\nu$ at four different levels, including the no-jumps case (i.e., the GBM case), as specified in the caption of Figure~\ref{fig:Hausdorff}. We vary the size of the change point and its position relative to the sample size as described in the $y$- and $x$-labels of Figure~\ref{fig:Hausdorff}.

The results presented in Figure~\ref{fig:Hausdorff} demonstrate that the LSTV$^*$(BV) method accurately estimates the true breakpoint and is robust against varying jump intensity levels and change point locations. Across most simulations, the Hausdorff distance is less than $0.1\%$ of the observations, which corresponds to less than three observations, even when the change point is of the order of $0.03$ and located at the $0.995$ percentile of the sample (i.e., with only $20$ observations following the change point).

Next, we compare the accuracy in detecting the locations of multiple change points for the proposed LSTV$^*$(BV) model against other multiple change point detection methods. Similarly to the simulations above, we consider the GBM and MJD processes with the aforementioned parameters. The number of true change points is set to $K^{*}=1,2,5,10$, and we perform $10000$ simulations per parameter combination. The spot volatility levels are the same as above and vary from 2 to 11 depending on the number of change points $K^{*}$. We utilized the previously determined  $\sigma$ values to generate a variety of computational simulations with randomly drawn volatility levels in different orders and magnitudes corresponding to the values used in the first simulation. These simulations encompass a broad spectrum of potential scenarios. Among these scenarios, we included the so-called staircase scenario as delineated by \cite{Rinaldo_Approximate_change_point_NIPS2017}. For the benchmarks, we consider two recently introduced methods, Tail Greedy Unbalanced Haar transformation (TGUH) by \cite{Fryzlewicz:18}, and locally refined change point (LRCP) estimator by \cite{WangRinaldo:20}, both have been implemented using the original R packages \textit{breakfast} and \textit{change points} developed by the authors of the respective papers. We are interested in the accuracy of the location of the detected (multiple) change points; hence, for LSTV$^*$(BV), we set $K_{max}=K^*$. Table ~\ref{tab:HausdorffError} summarizes average Hausdorff errors from 10,000 simulations as a percentage of the sample size. The LSTV$^*$(BV) outperforms the other approaches in five out of eight cases, and in the remaining cases, it performs similarly to the other methods. 

\begin{table}[htbp]
\scriptsize
  \centering
    \begin{tabular}{ccccccc}
    \toprule
      & \multicolumn{3}{c}{GBM} & \multicolumn{3}{c}{MJD} \\
    \midrule
    $K^{*}$ & LSTV$^*$(BV) & TGUH & LRCP & LSTV$^*$(BV) & TGUH & LRCP \\
    \midrule
    1 & 9.042 & 9.686 & \textbf{8.940} & \textbf{17.802} & 17.868 & 17.814 \\
    2 & \textbf{4.630} & 4.760 & 5.802 & 5.110 & \textbf{4.996} & 4.988 \\
    5 & \textbf{4.780} & 4.858 & 13.868 & \textbf{4.780} & 5.434 & 14.162 \\
    10 & \textbf{11.670} & 12.310 & 19.444 & 10.670 & \textbf{10.442} & 20.742 \\
    \bottomrule
    \end{tabular}%
  \caption{Comparison of Hausdorff error (as a percentage of the sample size) between LSTV$^*$(BV), TGUH, and LRCP methods for different values of $K^*$ (average across $10000$ simulations). }
  \label{tab:HausdorffError}
\end{table}

We now evaluate the proposed estimator on real financial data from the New York Stock Exchange (NYSE) from January 2018 through June 2022. We select the ten largest stocks in terms of average market capitalization based on data from NYSE Trade \& Quote (TAQ). Because the trading times for the stocks are not synchronized, and prices are subject to market microstructure noise and end-of-day effects, we use mid-prices and impute missing values with historical prices. We obtain cleaned returns aggregated to frequencies of $1$, $5$, $15$, $30$, and $60$ minutes during trading hours between 9:30 a.m.\ and 4:00 p.m. The number of daily observations varies from 6 to 390, depending on the frequency. We do not include overnight returns, as they are influenced by external factors.

For example, we plot the $1$-minute frequency squared returns of Tesla (TSLA) for two weeks of August 2022 in Figure~\ref{fig:TSLA}. Using the LSTV$^{*}$(BV) estimator with $K_{\mathrm{max}}=100$ and $\xi = \{ 0.05, 0.1, 0.2, 0.3 \}$, we find that although there are many jumps in the price process as can be seen from the isolated spikes in the squared log returns, the algorithm detects only significant change points. As anticipated, most of these change points occur around market open and close times, and as we increase the threshold value, the number of detected change points decreases. For example, in Figure~\ref{fig:TSLA}, $4$ change points are detected on day $7$ for $\delta=0.05$, whereas for $\delta=0.3$, only $1$ change point is detected. Interestingly, there is no significant difference in the detected change points between $\delta = 0.2$ and $\delta = 0.3$, suggesting that the algorithm exhibits low sensitivity to the threshold parameter. Ultimately, users should determine the appropriate $K_{\mathrm{max}}$ and $\delta$ based on data frequency and observed market regime changes tailored to the specific requirements of volatility prediction applications. In risk management and trading, this includes the forecasting horizon, level of transaction costs, market liquidity, frequency of rebalancing, and other market frictions. These quantities vary for different market participants, making the flexibility of choosing $K_{\mathrm{max}}$ and $\delta$ a practical feature of the proposed model. Moreover, as indicated by the estimates in Figure~\ref{fig:TSLA}, volatility levels are altered mostly only by significant and persistent events. This provides more stability in the predicted volatility and is important in practical implementation.

\begin{figure}
\begin{center}
\includegraphics[width=0.999\textwidth, height=0.65\textwidth]{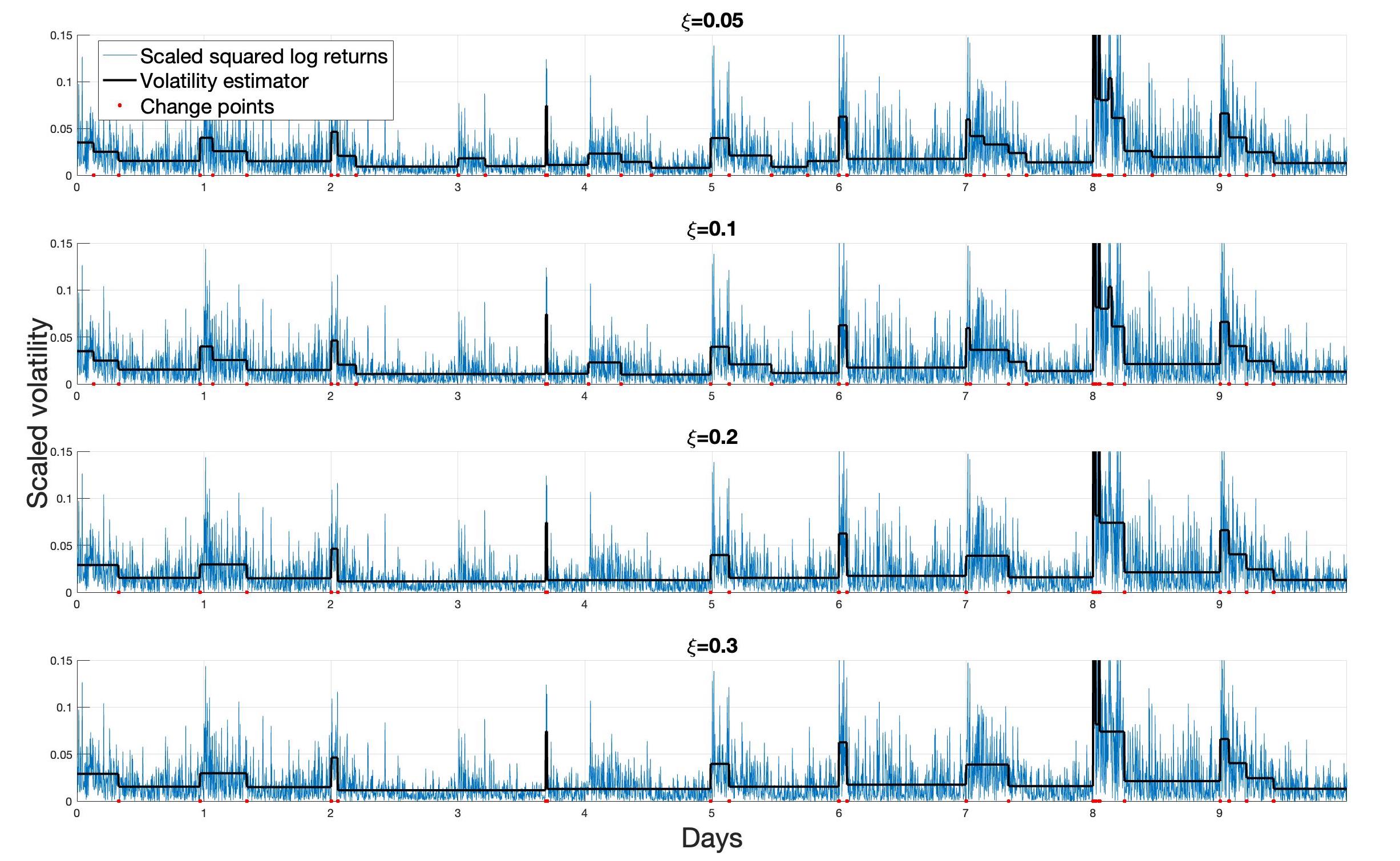}
\caption{
Tesla squared 1-minute returns over two weeks in August 2020, accompanied by the corresponding estimated volatilities using LSTV$^{*}$(BV) with a maximum change points $K_{\mathrm{max}}=100$, and threshold parameter $\xi=\{0.05, 0.1, 0.2, 0.3\}$.  The location of the detected change points is indicated with dots on the $x$-axis. Notably, an increase in the threshold ($\xi$) results in a decrease in the number of detected change points.
\normalsize}
\label{fig:TSLA}
\end{center}
\end{figure}

\begin{figure}[ht]
\begin{center}
\includegraphics[width=0.95\textwidth, height=0.75\textwidth]{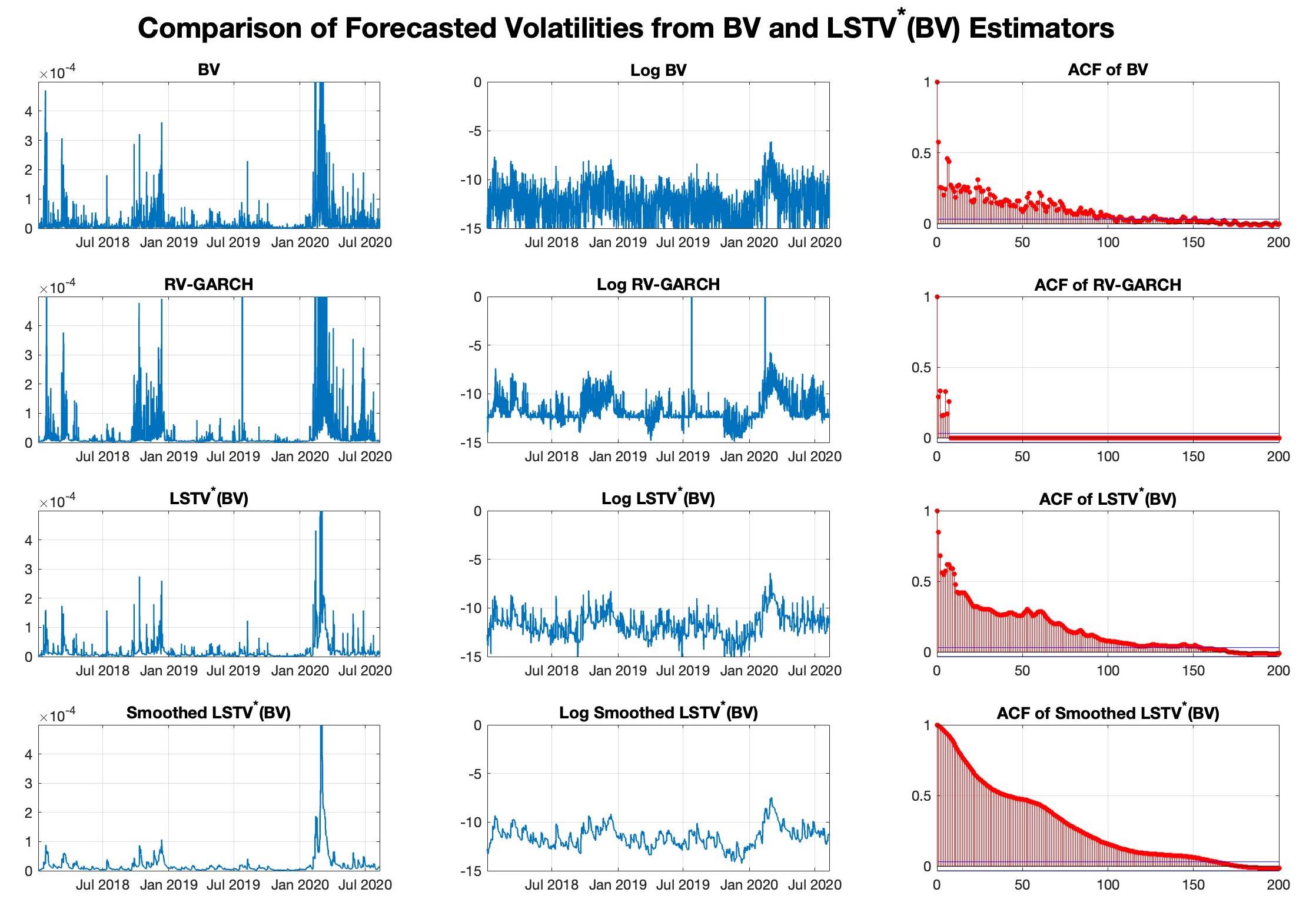}
\caption{One-step-ahead forecasts from the rolling window exercise for the Apple 60 min stock using two weeks of the rolling window from January 2018 to August 2020. Row-wise, comparing BV, realized RV-GARCH, LSTV$^*$(BV), and smoothed LSTV$^*$(BV) models. \textbf{First column:} forecasted volatility values. \textbf{Second column:} log of the volatilities. \textbf{Third column:} auto-correlation functions of the forecasted volatilities.
}
\label{fig:SmoothedComparison}
\end{center}
\end{figure}

Next, to assess the forecasting ability of the proposed estimators, we compare them to various benchmarks. We conduct a rolling window exercise, in which we estimate our models and generate one-step-ahead forecasts for different data frequencies. Specifically, for a given window of data consisting of observations $r_{1,n}, r_{2,n}, \cdots, r_{t,n}$, we aim to forecast the spot volatility term one-step-ahead matching the frequency of the input data, as well as the integrated variance over one-day-ahead period.

Formally, to calculate the forecasts, we use the falling factorial basis from \citet{Wang:14} and form the predictions at arbitrary points as derived in \citet{Ramdas:16}. We estimate the piecewise constant component of the volatility. Hence, the falling factorial basis prediction reduces to extrapolating the last estimated volatility $\widehat{\sigma}(t/n)$ obtained from the corresponding model within a given rolling window to $t+1$ for our high-frequency estimators. 

We consider a broad spectrum of volatility estimators commonly used in practice as benchmarks. Starting from the classical QV and BV estimators introduced by \citet{NielsenShephard:04}, we also include more recent refinements, such as the realized volatility and bipower variation (RVB) and realized volatility and quarticity (RVQ) estimators proposed by \citet{Yu:20}. In addition, we consider the heavy-tailed $t$-GARCH model with Student-$t$ innovations and its long-memory extension, the $t$-FIGARCH model, introduced by \citet{ShiHo15}. We also considered kernel-based volatility estimators introduced by \citet{shephard-designing-realised-kernels-2008}. Finally, we include the mixed-frequency approach of the HAR-RV model proposed by \citet{Corsi:09}. However, the HAR-RV model can only be a benchmark for multi-step (daily) volatility forecasts. Regarding truncation-based estimators mentioned in Section \ref{sec:TheoreticalResults}, we investigated the forecasting performance of MinRV/MedRV estimators proposed by \citet{andersen2012truncation}, and they performed worse than considered power estimators. Hence, they are omitted from the final comparison results.  

If there are no change points, the LSTV$^{*}$(QV) and LSTV$^{*}$(BV) estimators reduce to the original QV and BV estimators, respectively. The forecasting exercises below set $K_{\mathrm{max}}=1$ and $\xi=0.3$. Allowing for one change point improves forecasting performance across all models, assets, and frequencies considered below, $\xi=0.3$ is as in \citet{HarchaouiLeduc:10}. We also experimented with $K_{\mathrm{max}}>1$ and different $\xi$ values---higher $K_{\mathrm{max}}$ require careful selection of $\xi$. In general, the user should decide on the appropriate choice of $K_{\mathrm{max}}$ and $\xi$ based on market regimes that she observes in her data and specific applications of the estimated volatility.

In Figure~\ref{fig:SmoothedComparison}, we compare rolling window-based one-step-ahead forecasted volatilities and log forecasted volatilities of Apple's 60-minute frequency returns using BV, the realized RV-GARCH by \cite{hansen2011forecasting}, and LSTV$^*$(BV) models. The forecasted volatilities for the BV model exhibit many spikes across consecutive rolling windows, while the realized RV-GARCH produces forecasts with few spikes. By contrast, our proposed model provides much smoother and more realistic forecasts, as seen from their auto-correlation structure. We can use a simple exponential-moving average filter to smooth the forecasts further. The bottom plots in Figure~\ref{fig:SmoothedComparison} show the results of applying this filter, with the parameter $\lambda$ chosen to minimize the overall squared error. The proposed estimator, although piecewise constant in-sample, when looking at forecasts from consecutive rolling windows, preserves more stylized facts of a typical volatility process out-of-sample than more flexible parametric models encompassing the same stylized facts in their dynamics in-sample. If volatility forecasting is the main goal, preserving stylized facts, out-of-sample arguably should be considered more important than the flexibility of the in-sample fit, and our parsimonious model allows us to achieve that.

To compare the forecast performance of different models quantitatively in our rolling window exercise, we use the prediction Average Squared Error (ASE) following the approach of \citet{Fryzlewicz:06},
\begin{align}
 ASE_{f,N}(\mathcal{M}) &= \frac{1}{N} \sum_{t=1}^{N}\left(\bar{\sigma}^2_{t+f|t}(\mathcal{M}) -\sum_{s=1}^fr^2_{t+s}\right)^2\label{eq:ASEdef} 
\end{align}
where $N$ is the number of rolling windows, $f$ is the forecast horizon in each window, $\bar{\sigma}^2_{t+f|t}$ denotes the forecasted integrated variance from $t+1$ to $t+f$ using model $\mathcal{M}$, and $r^2_{t+s}$ are the squared returns at time $t+s$. In order to compare a given model $\mathcal{M}_1$ to a benchmark $\mathcal{M}_2$, we report percentage improvement results computed as $$100\frac{ASE_{f,N}(\mathcal{M}_1)-ASE_{f,N}(\mathcal{M}_2)}{ASE_{f,N}(\mathcal{M}_2)}$$

In the left panels of Figure~\ref{fig:AAEASE}, we compare the forecasting ability of our proposed methods to the classical QV and BV estimators by \citet{NielsenShephard:04}, as well as the more recent refinements of them, including the realized volatility and bipower variation (RVB) and realized volatility and quarticity (RVQ) estimators proposed by \citet{Yu:20} and realized kernel (RK) estimator introduced by \cite{shephard-designing-realised-kernels-2008}. The top panel of Figure~\ref{fig:AAEASE} includes results for the one-step-ahead forecast. Daily volatility forecasts are widely used in trading, e.g., for volatility-adjusted position sizing, and in risk management, e.g., for daily Value-at-Risk forecasts. The bottom panel of Figure~\ref{fig:AAEASE} shows results for the one-day-ahead forecast. The daily forecast also includes the HAR-RV model by \citet{Corsi:09}.
All the percentage improvements are relative to the BV model as the benchmark.
The LSTV$^*$(QV) outperforms the original QV, but it is still worse than the jump-robust BV. Among all the models, data frequencies, and forecast horizons, only the LSTV$^*$(BV) outperforms the original BV method with a significant improvement of around $5\%$ and $20\%$ for one-step-ahead and one-day-ahead forecasts, respectively. In the daily forecast, even the HAR-RV model is outperformed by our approach.

In Figure~\ref{fig:FIGARCH}, we compare our approach with GARCH models, which capture better the in-sample dynamics of the volatility, including the original GARCH(1,1), $t$-GARCH(1,1), and the more recent long-memory variation of GARCH model---FIGARCH(1,1)---proposed by \citet{ShiHo15}. As a benchmark, we again used the BV estimator. Our original  LSTV$^*$(QV) and LSTV$^*$(BV) estimators capture the piecewise constant component of the volatility; in case more accurate in-sample volatility dynamics are of interest, we also introduce a variation of our estimator by using the estimated volatility from LSTV$^*$(BV) as innovations for the FIGARCH model. This is referred to as LSTV$^*$(FIGARCH). We find that both FIGARCH and LSTV$^*$(FIGARCH) perform worse than the benchmark BV, while LSTV$^*$(BV) outperforms all other methods across all frequencies of the data. The LSTV$^*$(FIGARCH) performs on par or better than the original FIGARCH, but it is outperformed by the BV and LSTV$^*$(BV) estimators.

\begin{figure}
    \centering
    {{\includegraphics[width=.99\linewidth, height=0.65\textwidth]{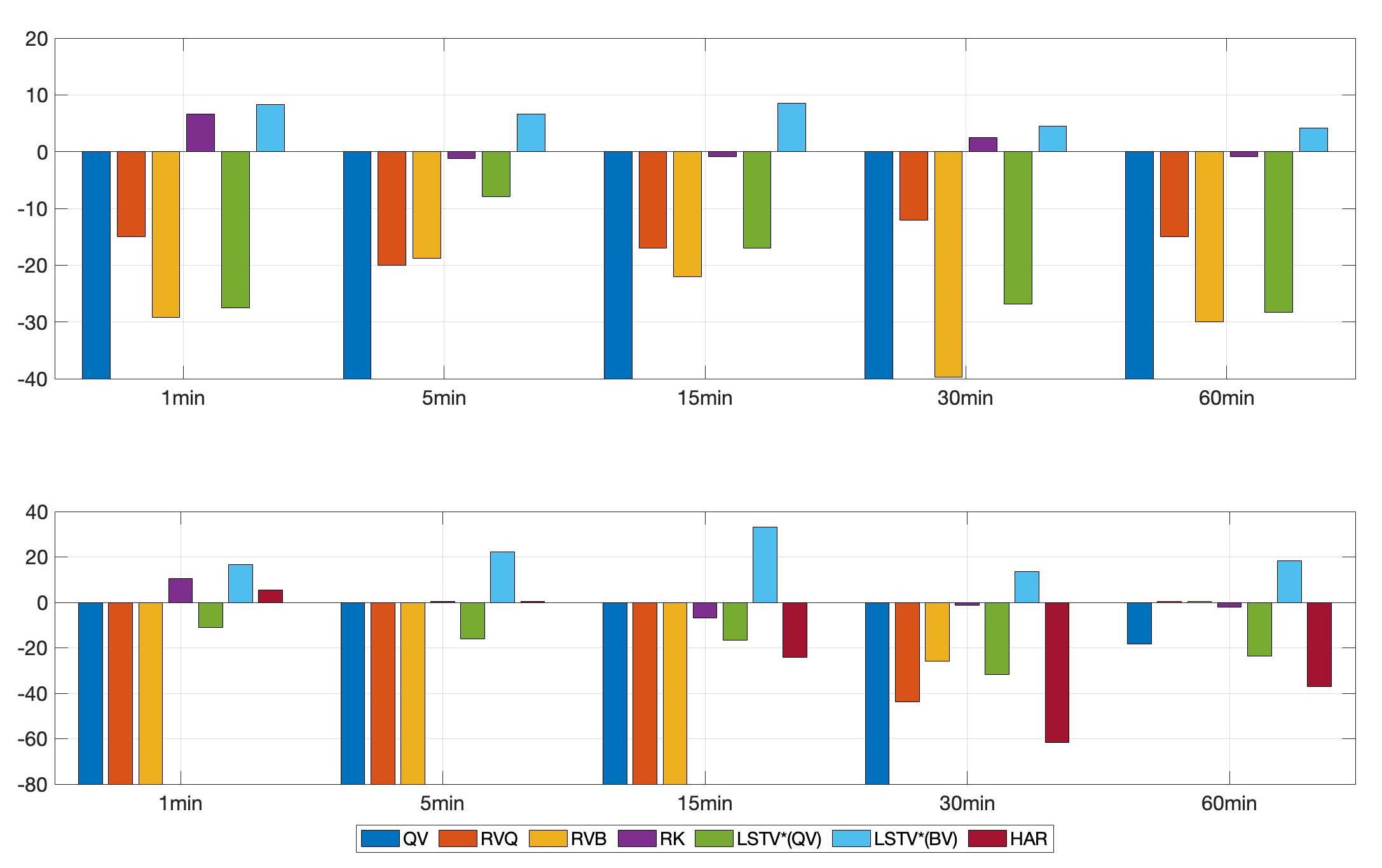} }}\hfill%
    \caption{Percentage improvements of volatility forecasts for different frequencies of data in terms of predicted ASE from \eqref{eq:ASEdef} relative to the performance of the BV estimator (negative values are truncated).  \textbf{Top  Panel:} One-step-ahead forecast. \textbf{Bottom  Panel:} One-day-ahead forecast. }
\label{fig:AAEASE}%
\end{figure}

\begin{figure}[ht]
    \centering
    \begin{subfigure}[t]{0.45\textwidth}\includegraphics[width=0.9\textwidth, height=\textwidth]{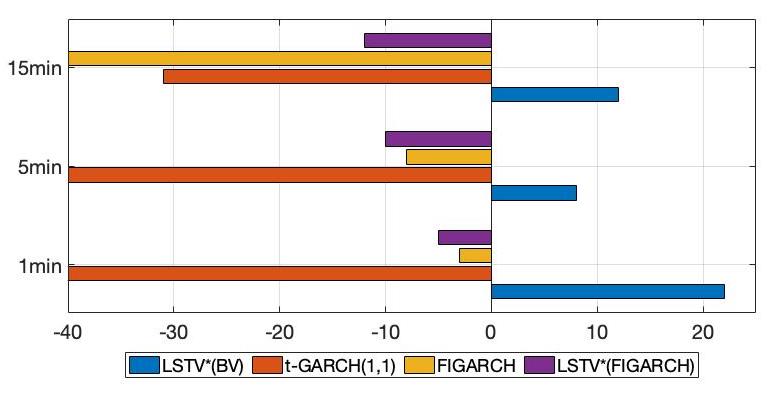}
        \caption{\scriptsize Figure 6a: $LSTV^*(BV)$ vs. GARCH-based models $t$-GARCH, FIGARCH, and LSTV$^*$(FIGARCH), all relative to the BV estimator (30-min and 60-min omitted because of poor performance of all the GARCH based models). Only the proposed $LSTV^*(BV)$ outperforms the original BV estimator in all panels and frequencies.}\label{fig:FIGARCH}
    \end{subfigure} \quad
    \begin{subfigure}[t]{0.45\textwidth}.\includegraphics[width=0.9\textwidth,
        height=\textwidth]{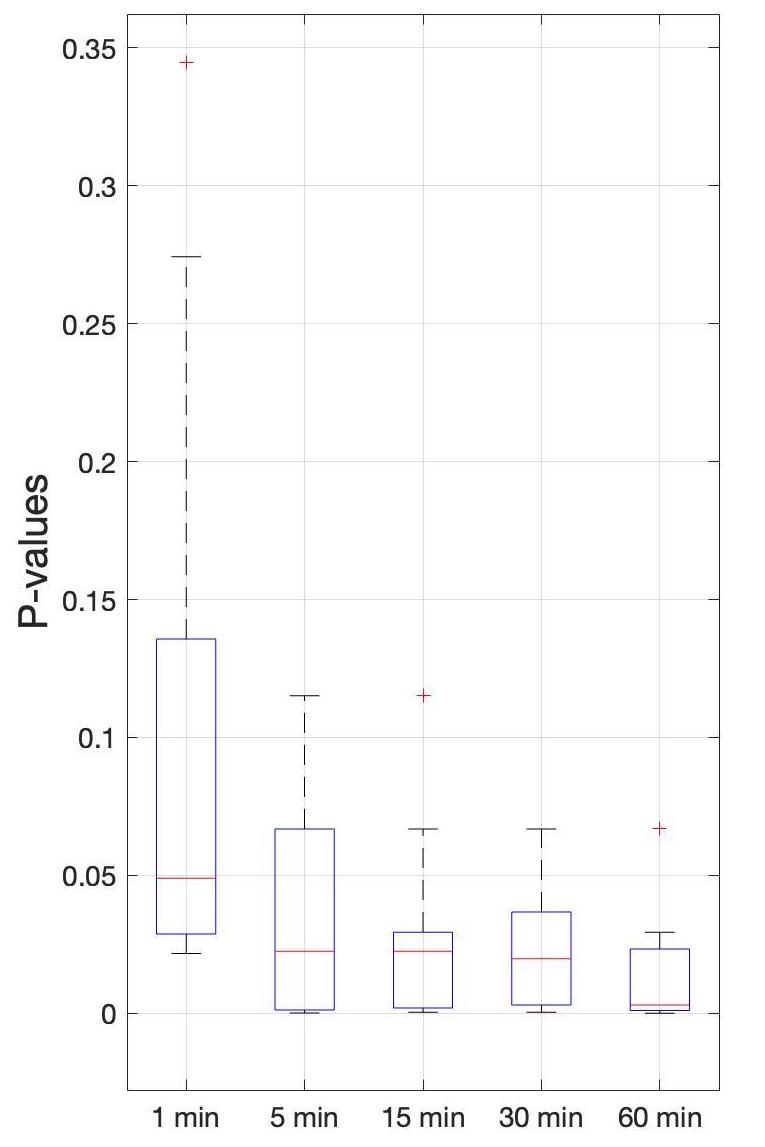}
        \caption{\scriptsize Figure 6b: Boxplots of Diebold-Mariano (DM) $p$-values for the ten stocks considered in the Figure \ref{fig:AAEASE}. The LSTV$^*$(BV) forecasts are compared against the benchmark model BV. Notably, LSTV$^*$(BV) demonstrates significant outperformance over BV, with $p$-values falling below 0.05 for most tickers across 5-min, 15-min, 30-min, and 60-min and for several tickers at the 1-min frequency.}\label{fig:DM}
    \end{subfigure}
\end{figure}

To assess the forecasting performance of our  LSTV$^*$(BV) model with respect to other models, we use the Diebold-Mariano (DM) test statistic from \citet{DieboldMariano2002}.  We use the forecasts from the exercise in Figure \ref{fig:AAEASE}, and calculate our model's DM test statistic $p$-values against benchmark model BV as shown in Figure~\ref{fig:DM}. LSTV$^*$(BV) models are significantly better than BV in frequencies 5-min, 15-min, 30-min, and 60-min with $p$-values$<0.05$. This shows that  LSTV$^*$(BV) outperforms BV significantly in terms of forecast accuracy. Whereas in 1-min frequency, LSTV$^*$(BV) outperforms BV with a chosen significance level of $0.1$.

The theoretical limitation of the total-variation-based change points detection methods compared with the TGUH approach is the slower convergence rate in the weak sparsity case, i.e., when the number of change points goes to infinity together with the sample size, see \citet{Rinaldo_Approximate_change_point_NIPS2017} and references therein. Nevertheless, as shown in  \citet{lin2016approximate}, the LSTV method achieves optimal convergence rate in strong sparsity cases when the number of change points does not grow with the number of observations. Since we are working with high-frequency data and infill-asymptotics, even when the number of observations goes to infinity, they cover the same time interval; hence, the number of change points should remain the same. Hence, assuming constant change points on the fixed time interval can be considered realistic. This is supported by the following forecasting exercise in which we compare our LSTV$^*$(BV) with the TUGH method from \citet{Fryzlewicz:18}. Table \ref{tab:ASEmultiplechangepts} summarizes the results regarding the relative percentage improvement of predicted ASE from \eqref{eq:ASEdef}. Using a rolling window of two weeks of data with different frequencies (1 min, 5 min, 15 min, 30 min, and 60 min), we detect change points using the two methods, estimate the corresponding volatilities, and provide one-day-ahead forecasts.  We use our LSTV$^*$(BV) with the maximal number of allowed change points $K_{max}\in \left\{1,5,10,50\right\}$ and the threshold parameter $\xi \in \left\{0.1,0.2,0.3,0.4,0.5,0.6,0.7,0.8,0.9\right\}$, respectively. Our method performs better than TGUH in most cases across frequencies and independently of $K_{max}$ and $\xi$ parameters. We perform best in the case of $K_{max}=1,5,10$, which is usually the case for two weeks of high-frequency data.
The observed improvement remains consistent across the entire threshold range, indicating that the forecast is relatively insensitive to variations in the additional threshold parameter. The choice of $K_{max}$ can be left to the user's discretion, considering factors such as the training window, characteristics of the assets under consideration, and overall market conditions.

\begin{table}[htbp]
\scriptsize
  \centering
    \begin{tabular}{ccccccccccc}
    \toprule 
\multicolumn{2}{c}{} & \multicolumn{9}{c}{Threshold $\xi$} \\
    Freq & Kmax & 0.1 & 0.2 & 0.3 & 0.4 & 0.5 & 0.6 & 0.7 & 0.8 & 0.9 \\
    \midrule
    {1 min} & 1 & -7.75 & -7.75 & -7.75 & -7.75 & -7.75 & -7.75 & -7.75 & -7.75 & -7.75 \\
      & 5 & \textbf{29.93} & \textbf{30.67} & \textbf{30.64} & \textbf{30.67} & \textbf{31.05} & \textbf{31.48} & \textbf{31.38} & \textbf{31.27} & \textbf{32.67} \\
      & 10 & \textbf{19.24} & \textbf{20.33} & \textbf{20.40} & \textbf{20.28} & \textbf{22.49} & \textbf{23.93} & \textbf{24.72} & \textbf{26.12} & \textbf{26.60} \\
      & 50 & -14.17 & -5.99 & -1.61 & -1.59 & \textbf{0.24} & \textbf{2.86} & \textbf{3.39} & \textbf{3.63} & \textbf{5.65} \\
      \midrule
    {5 min} & 1 & \textbf{44.06} & \textbf{44.06} & \textbf{44.06} & \textbf{44.06} & \textbf{44.06} & \textbf{44.06} & \textbf{44.06} & \textbf{44.06} & \textbf{44.06} \\
      & 5 & -13.62 & -12.66 & -11.85 & -11.81 & -11.11 & -10.93 & -9.46 & -8.17 & -7.45 \\
      & 10 & -13.26 & -13.17 & -12.26 & -11.56 & -11.70 & -11.34 & -11.46 & -11.13 & -11.80 \\
      & 50 & \textbf{3.93} & \textbf{4.67} & \textbf{6.20} & \textbf{5.31} & \textbf{4.01} & \textbf{3.56} & \textbf{3.43} & \textbf{3.88} & \textbf{3.97} \\
      \midrule
    {15 min} & 1 & \textbf{10.54} & \textbf{10.54} & \textbf{10.54} & \textbf{10.54} & \textbf{10.54} & \textbf{10.54} & \textbf{10.54} & \textbf{10.54} & \textbf{10.54} \\
      & 5 & -2.59 & -2.61 & -1.48 & \textbf{7.93} & \textbf{8.38} & \textbf{8.83} & \textbf{7.84} & \textbf{8.23} & \textbf{6.47} \\
      & 10 & -0.02 & \textbf{0.76} & -0.27 & -2.87 & -1.58 & -2.84 & -3.17 & -2.97 & -3.49 \\
      & 50 & \textbf{5.58} & \textbf{6.73} & \textbf{8.39} & \textbf{9.87} & \textbf{10.85} & \textbf{10.11} & \textbf{9.65} & \textbf{10.51} & \textbf{9.95} \\
      \midrule
    {30 min} & 1 & -3.00 & -3.00 & -3.00 & -3.00 & -3.00 & -3.00 & -3.00 & -3.00 & -3.00 \\
      & 5 & \textbf{7.16} & \textbf{7.16} & \textbf{7.85} & \textbf{8.22} & \textbf{8.22} & \textbf{9.18} & \textbf{9.90} & \textbf{8.55} & \textbf{9.91} \\
      & 10 & \textbf{5.15} & -2.16 & -2.87 & -1.18 & \textbf{1.04} & \textbf{2.41} & \textbf{2.90} & \textbf{4.04} & \textbf{4.72} \\
      & 50 & -4.34 & -1.14 & \textbf{1.71} & \textbf{5.08} & \textbf{5.14} & \textbf{8.19} & \textbf{8.31} & \textbf{8.38} & \textbf{8.31} \\
      \midrule
    {60 min} & 1 & \textbf{88.03} & \textbf{88.03} & \textbf{88.03} & \textbf{88.03} & \textbf{88.03} & \textbf{88.03} & \textbf{88.03} & \textbf{88.03} & \textbf{88.03} \\
      & 5 & \textbf{89.62} & \textbf{89.55} & \textbf{89.44} & \textbf{89.38} & \textbf{89.38} & \textbf{89.44} & \textbf{89.42} & \textbf{89.46} & \textbf{89.66} \\
      & 10 & \textbf{88.69} & \textbf{88.81} & \textbf{88.92} & \textbf{89.03} & \textbf{89.10} & \textbf{89.12} & \textbf{89.16} & \textbf{89.31} & \textbf{89.26} \\
      & 50 & \textbf{89.54} & \textbf{89.63} & \textbf{89.87} & \textbf{89.94} & \textbf{90.19} & \textbf{90.12} & \textbf{90.14} & \textbf{90.25} & \textbf{90.37} \\
    \bottomrule
    \end{tabular}%
    \caption{Percentage improvements of one-day-ahead forecasting in terms of the Average Squared Error (ASE) as in \eqref{eq:ASEdef} from LSTV$^*$(BV) method with respect to Tail Greedy Unbalanced Haar transformation (TGUH) by \cite{Fryzlewicz:18} as a benchmark.}
  \label{tab:ASEmultiplechangepts}%

\end{table}%

\section{Extensions and Further Research}\label{sec:Extensions}
In addition to their performance in high-frequency volatility estimation, our proposed volatility estimators can be extended to capture other stylized facts of high-frequency volatility, demonstrating their versatility and practical applicability. For example, one can filter jumps from the stock prices separately from the volatility estimation by adding the $\ell_1$ regularization term to the objective function in \eqref{equation:lasso-filter-optimization}.

Similarly, intraday or weekly seasonality, a well-known statistical characteristic of high-frequency stock volatility, can be incorporated into our $\ell_1$ volatility filter. One approach for incorporating seasonality is to follow \citet{vatter2015non} and incorporate smooth trend and seasonality into the framework. To capture both effects, the objective function \eqref{equation:lasso-filter-optimization} can be modified as follows:
\begin{align*}
 \dfrac{1}{n}\norm{\tilde{\vartheta}_t  -\mathbf{X}\boldsymbol\beta - \mathbf{s} - \mathbf{v}}^2_2 + \lambda\norm{\boldsymbol\beta}_1 + \gamma\norm{\mathbf{v}}_1\\
\textrm{s.t.} \quad  s_{t+p} = s_T  \quad   t=1,\ldots,n-p,
\end{align*}
where $\mathbf{s}=(s_1,\ldots,s_n)'\in \mathbb{R}^n$ is the seasonal periodic component with known period $p$, vector $\mathbf{v} \in \mathbb{R}^n$ is the sparse spikes component, and the parameter $\gamma \geq 0$ controls jumps intensity. 

If the exact seasonality pattern needs to be estimated from the data, one can use
\begin{align*}
 &\dfrac{1}{n}\norm{\tilde{\vartheta}_t -\mathbf{X}\boldsymbol\beta - \mathbf{s} - \mathbf{v}}^2_2  + \lambda\norm{\boldsymbol\beta}_1 + \gamma\norm{\mathbf{v}}_1 +\sum_{\omega\in \Omega}\left(\left|a_\omega\right|+\left|b_\omega\right|\right)\\
&\textrm{s.t.} \quad  s_{t} = \sum_{\omega \in \Omega}a_\omega \sin(\omega t) + b_{\omega} \cos(\omega t)\quad t=1,\ldots,n. 
\end{align*}

The $\mathbf{s}=(s_1,\ldots,s_n)'\in \mathbb{R}^n$ filter is responsible for selecting the best frequencies from an over-complete dictionary $\Omega$. In this case, the frequency does not need to be known a priori, and the penalty term discourages greedily using the dictionary. For an adaptive version of this filter and a fast coordinate descent algorithm for estimation, see \citet{SoutoGarcia:16}.

Secondly, one of the assumptions in our analysis so far is that the stock returns are equally spaced because the data used in Section \ref{sec:Empirics} consists of TAQ data aggregated to 1, 5, 10, 30, and 60-minute intervals. We can extend our method to tick-level data, which is irregularly spaced. First, consider the estimation of the underlying continuous volatility from a finite number of data points. Equation \eqref{equation:lasso-filter-optimization} can be extended to a continuous form as follows:
\begin{equation*}
 \dfrac{1}{n}\sum_{i=1}^n (\vartheta_i - x_{t_i})^2 + \lambda \int_{t_1}^{t_n} \abs{\dot{x}}\mbox{d}t.  
\end{equation*}
This can be simplified by using the standard interpolation problem
as explained in \cite{KimBoyd:08}. In particular, we can solve
\begin{equation*}
 \dfrac{1}{n}\sum_{i=1}^n \left( \tilde{\vartheta}_i - x_{t_i}\right)^2 + \lambda \sum_{i=2}^n \abs{\frac{x_{t_{i+1}}-x_{t_i}}{t_{i+1}-t_i}}.
\end{equation*}

Using the optimal points $x_{t_i}$, $i=1,\ldots,n$, we can recover the solution to the original piecewise linear trend filtering problem: the piecewise-linear function, $x^*_T$, for $x^*_{t_i}$, given by
$$x^*_T = \frac{t_{i+1}-t}{t_{i+1}-t_{i}} x^*_{i},$$
where $i=1,\ldots,n$, is the optimal piecewise constant volatility. To reflect the non-equally spaced data, one can modify the proposed LSTV$^*$ algorithm by modifying the matrix $\mathbf{X}$ to include the time increments from the regularization term. Alternatively, one can use the modified version of the primal-dual algorithm described in \citet{KimBoyd:08}.

In this paper, we assumed that the observed returns are the true returns. In ultra-high-frequency data, one needs to accommodate the assumption of microstructure noise. We can do that by using a realized kernel estimator $F_\mathfrak{K}(\mathbf{r}_T)$, where $\mathfrak{K}$ stands for the specific kernel function (see \citealp{shephard-designing-realised-kernels-2008, shephard-realised-kernels-2009, shephard-multivariate-realised-kernels-2011, shephard-subsampling-realised-kernels-2011}, and \citealp{shephard-realized-kernels-2015}). Indeed, if we have $p(t) = p^{*}(t) + \epsilon(t),$
where $p(t)$ denotes the observed price at $t$, $p^{*}(t)$ denotes the true price at $t$, and $\epsilon$ denotes the micro-structure noise; then the returns satisfy $r_T = r^{*}_T + u_T,$ where $u_T = \epsilon(t) - \epsilon(t-1)$, and one can construct a kernel-based estimator $F_\mathfrak{K}$ that satisfies
$$\plim_{n \rightarrow +\infty} F_\mathfrak{K}(\mathbf{r}_T) = [\mathbf{r}^{*}_T] \mathrm{ and } \plim_{n \rightarrow +\infty} = F_\mathfrak{K}(\mathbf{u}_T) = 0,$$
where $\mathbf{r}_T$ is the vector of high frequency returns; as before, $[\mathbf{r}^{*}_T]$ is the quadratic variation of the true returns; and $\mathbf{u}_T = (u_1, \cdots, u_n)$.

Finally, the estimators proposed in this paper allow the construction of an estimator of the jump component. Namely, one can take the differences  between increments of the QV and BV estimator and plug them into the objective function with $\ell_1$ regularization not on the difference but on the series itself, i.e.,
\begin{align}
\widehat{\mathbf{u}}(\lambda) &:= \argmin_{\mathbf{u} \in \mathbb{R}^{n-1}} \mathcal{L}(\mathbf{u}) 
\label{eq:objective_feasible_jumps}\\ 
\mbox{where }\mathcal{L}(\mathbf{u}) &:= \frac{1}{n }\norm{\tilde{F}_{\mathrm{RV}}(\mathbf{r}_T) - F_{\mathrm{BP}}(\mathbf{r}_T) -\mathbf{u}}^2_{2} + \lambda\norm{\mathbf{u}}_{1}, \nonumber
\end{align}
where $\tilde{F}_{\mathrm{RV}}$ is the original ${F}_{\mathrm{RV}}$ with the first element $r_{1,t}^2$ in the vector removed.
The $\widehat{\mathbf{u}}(\lambda)$  in \eqref{eq:objective_feasible_jumps} estimates the jump component of the price process. They can be obtained via the original LARS or any other LASSO-related algorithm. Since $\tilde{F}_{\mathrm{RV}}(\mathbf{r}_T) - F_{\mathrm{BP}}(\mathbf{r}_T)$ satisfies our consistency condition for the jump component, the results in this paper also show that $\widehat{\mathbf{u}}(\lambda)$ will be consistent for the true jump sizes and their locations.

\section{Conclusions} \label{sec:Conclusion}
In this paper, we have introduced a new family of high-frequency volatility estimators. By leveraging the LSTV change point detection framework and adding $\ell_1$ regularization to two classical estimators (quadratic variation and bipower variation), we have proposed consistent estimators for volatility estimation and breakpoint detection. Simulation results demonstrate that the proposed estimators accurately identify breakpoints and outperform classical and recent high-frequency volatility estimators in out-of-sample volatility prediction at all frequencies and forecasting horizons.

Our proposed estimators offer practical advantages, such as their ability to accurately identify breakpoints close to the end of the sample and provide more realistic and smoother volatility forecasts that accurately capture the stylized facts of volatility when looking at consecutive out-of-sample predictions. In future research, we plan to extend the proposed estimators to the multivariate case and incorporate other stylized facts of (ultra) high-frequency financial returns using the extensions discussed in Section \ref{sec:Extensions}.

\bibliographystyle{plainnat}
\setcitestyle{round,authoryear,semicolon}
\small
\bibliography{MathematicalFinance}

\newpage

\Large \textbf{Appendix} \normalsize

\section{Sub-Weibull random variables}\label{appendix:subweibull}
A random variable $X$ with sub-Weibull distribution with tail parameter $\alpha$ is one such that
$$\forall t > 0 : \mathbb{P}\left(\abs{X-\mathbb{E}[X]} \geq t\right) \leq 2\exp\left(-(t/C)^{\alpha}\right)$$
for some $\alpha, C > 0$.
The parameters $\alpha = 1$ and $\alpha = 2$ correspond to sub-exponential and sub-Gaussian random variables, respectively. This class of heavy-tailed distributions (and a corresponding notion for random vectors) has recently been defined and studied by \citet{kuchibhotla2022moving}. (See also \citet{zhang-subweibull-2022}.) Concentration inequalities for such random variables (and the corresponding random vectors) have also been recently derived by \citet{gotze-sambale-alpha-concentration-2021} and \citet{sambale-notes-subalpha-2023}. Many of these results have been extended to a larger class of heavy-tailed distributions by \citet{sharp-concentration-heavy-tailed-maleki-2023}.

We recall a couple of properties of sub-Weibull distributions. 
Note that an equivalent definition of the class $\mathrm{subWeib}(\alpha)$ can be stated in terms of Orlicz-(quasi-)norms. Indeed, $X \in \mathrm{subWeib}(\alpha)$ if and only if
$$\norm{X}_{\Psi_\alpha} := \inf\left\{\eta > 0 : \mathbb{E}\left[\exp\left(\left(X/\eta\right)^{\alpha}\right)\right] \leq 2\right\}$$
is finite. Here, $\Psi_\alpha$ denotes the Orlicz function $s \mapsto \exp(s^{\alpha})-1$, $\alpha > 0$.\\

Note that in general, the constant $C$ in the tail bound definition of a sub-Weibull random variable $X$ and $\norm{X}_{\Psi_{\alpha}}$ are proportional to each other by a positive factor depending only on $\alpha$. (See \citet{sambale-notes-subalpha-2023} and \citet{gotze-sambale-alpha-concentration-2021}.)

\begin{proposition}
Let $X_1$ and $X_2$ be sub-Weibull random variables with tail parameters $\alpha_1$ and $\alpha_2$, respectively. Then, $X_1 X_2$ and $X_1 + X_2$ are sub-Weibull with respective tail parameters $\alpha_1 + \alpha_2$ and $\max(\alpha_1, \alpha_2)$.
\end{proposition}

\begin{proposition}\label{prop-subweibull-concentration}
Let the random variables $X_1,\cdots, X_n \in \mathrm{subWeib}(\alpha)$, with $\alpha > 0$, be independent, and suppose that $\norm{X_i}_{\Psi_{\alpha}} \leq M$ for all $i = 1, \cdots, n$. Then, for all $n \geq 1$ and $t > 0$, we have
$$\mathbb{P}\left(\dfrac{1}{n}\abs{\sum_{i=1}^n \left(X_i-\mathbb{E}(X_i)\right)} \geq t\right) \leq 2\exp\left(-\dfrac{1}{C_{\alpha}}\cdot\dfrac{t^{\alpha} n^{\alpha/2}}{M^{\alpha}}\right)$$
where $C_{\alpha} > 0$ is a universal constant depending only on $\alpha$.
\end{proposition}

This follows from the proof of \citep[Lemma 5.2]{sambale-notes-subalpha-2023}.

\section{Technical Lemmas}\label{appendix:lemmas}
\begin{lemma}\label{lemma:KKT-change points}
Let $\left(\widehat{t}_1, \cdots, \widehat{t}_n\right)$ denote the LASSO change point estimates, and let $\left(\hat{\vartheta}_1, \cdots, \hat{\vartheta}_n\right)$ denote the LASSO estimator of $F(\mathbf{r}_T)$. Then $\left(\widehat{t}_1, \cdots, \widehat{t}_n\right)$ and $\left(\hat{\vartheta}_1, \cdots, \hat{\vartheta}_n\right)$ satisfy
\begin{equation}
\forall \ell \in \left\{1, \cdots, \abs{\widehat{\mathcal{A}}(\lambda)}\right\}: \sum_{i = \widehat{t}_{\ell}}^n \vartheta_i - \sum_{i = \widehat{t}_{\ell}}^n \hat{\vartheta}_i = \dfrac{n \lambda}{2}\mathrm{sign}\left(\hat{\vartheta}_{\widehat{t}_\ell(\lambda)} - \hat{\vartheta}_{\widehat{t}_\ell(\lambda)-1}\right), \mathrm{ and }
\end{equation}

\begin{equation}
\forall j \in \left\{1, \cdots, n\right\} : \abs{\sum_{i = j}^n \vartheta_i - \sum_{i = j}^n \hat{\vartheta}_i} \leq \dfrac{n\lambda}{2}.
\end{equation}

The vector $\left(\hat{\vartheta}_{1}(\lambda),\cdots, \hat{\vartheta}_{n}(\lambda)\right)$ satisfies:
\begin{equation}
    \hat{\vartheta}_j(\lambda) = \widehat{\eta}_k \mathrm{ for } \widehat{t}_{k-1}(\lambda) \leq j \leq \widehat{t}_{k}(\lambda)-1,
\end{equation}
for some constant values $\widehat{\eta}_1, \cdots, \widehat{\eta}_{\widehat{K}}$, where $\widehat{K} \geq 1$.
\end{lemma}

For a proof, see \citep{HarchaouiLeduc:10}. The proof relies on a minimization criterion based on sub-gradients for each one of the objective functions.

In what follows, let $\mathbf{Z}[r,s] := \sum_{i = r}^s Z_i$ for any random vector $\mathbf{Z} = (Z_1, \cdots, Z_n)$.

\begin{lemma}\label{lemma:partial-sum-noise-variables}
Let $\zeta^{J}_1, \cdots, \zeta^{J}_n$ be $\mathrm{subWeib}(\alpha)$ random variables representing the compound Poisson jump components. If $\{\nu_n\}_{n \geq 1}$ and $\{\mu_n\}_{n \geq 1}$ are two positive sequences such that $\mu^{\alpha}_n \nu_n^{\alpha}\left[\log(n)\right]^{-1} \xrightarrow[n \to +\infty]{} +\infty$, then
$$\mathbb{P}\left(\sup_{\substack{1 \leq r_n < s_n \leq n \\ \abs{r_n-s_n} \geq \zeta_n } } \abs{\dfrac{\boldsymbol{\zeta}^J[r_n,s_n-1]}{s_n-r_n}} \geq \mu_n\right) \xrightarrow[n \to +\infty]{} 0.$$
\end{lemma}

\begin{proof}
Clearly,
$$\mathbb{P}\left(\sup_{\substack{1 \leq r_n < s_n \leq n \\ \abs{r_n-s_n} \geq \zeta_n }} \abs{\dfrac{\boldsymbol{\zeta}^J[r_n,s_n-1]}{s_n-r_n}} \geq \mu_n\right) \leq \sum_{\substack{1 \leq r_n < s_n \leq n \\ \abs{r_n-s_n} \geq \zeta_n }} \mathbb{P}\left(\abs{\dfrac{\boldsymbol{\zeta}^J[r_n,s_n-1]}{s_n-r_n}} \geq \mu_n\right).$$
Next, we have:
\begin{align*}
\mathbb{P}\left(\abs{\dfrac{\boldsymbol{\zeta}^J[r_n,s_n-1]}{s_n-r_n}} \geq \mu_n\right) &= \mathbb{P}\left(\abs{\sum_{j = N_{\tau_{r_n-1,t}}+1}^{N_{\tau_{s_n-1,t}}} Q^2_j} \geq \mu_n\abs{s_n-r_n}\right)\\
    &= \mathbb{P}\left(\abs{\sum_{j = 1}^{N_{\tau_{s_n-1,t}}} Q^2_j - \sum_{j = 1}^{N_{\tau_{r_n-1,t}}} Q^2_j} \geq \mu_n\abs{s_n-r_n}\right)\\
    &\leq \mathbb{P}\left(\abs{\sum_{j = 1}^{N_{\tau_{s_n-1,t}}} Q^2_j} \geq \dfrac{1}{2} \mu_n\abs{s_n-r_n}\right)\\
    &\quad + \mathbb{P}\left(\abs{\sum_{j = 1}^{N_{\tau_{r_n-1,t}}} Q^2_j} \geq \dfrac{1}{2} \mu_n\abs{s_n-r_n}\right)\\
    &\leq 4 \exp\left(-C_{\alpha}\mu_n^{\alpha}\abs{s_n-r_n}^{\alpha}\right),
\end{align*}
where $C_{\alpha}$ is an absolute constant depending on $\alpha$, $\norm{Q}_{\Psi_{\alpha}}$, and $T$. Summing over $1 < r_n < s_n \leq n$ with $\abs{r_n-s_n} \geq \nu_n$, we obtain
\begin{align*}
    \mathbb{P}\left(\sup_{\substack{1 \leq r_n < s_n \leq n \\ \abs{r_n-s_n} \geq \zeta_n }} \abs{\dfrac{\boldsymbol{\zeta}^J[r_n,s_n-1]}{s_n-r_n}} \geq \mu_n\right) &\leq 4\left(\dfrac{n(n-1)}{2}\right)\exp\left(-C_{\alpha}\mu_n^{\alpha}\zeta^{\alpha}_n\right)\\
    &\leq 2 \exp\left[2\log(n)\left(1-C_{\alpha}\dfrac{\mu^{\alpha}_n \nu_n^{\alpha}}{\log(n)}\right)\right],
\end{align*}
and the latter upper bound converges to $0$ as $\mu^{\alpha}_n \zeta_n^{\alpha}\left[\log(n)\right]^{-1} \xrightarrow[n \to +\infty]{} +\infty$.
\end{proof}

\section{Proofs}\label{app:Proofs}

In this section, we present the proof of the results outlined in the paper.
\begin{proof}[Proposition \ref{proposition:mse-consistency}]
By Abel's summation formula:
$$\sum_{i = 1}^n (\tilde{z}_i - z_i)^2 = \sum_{i = 1}^n (\tilde{z}_i + z_i)(\tilde{z}_i - z_i) = (\tilde{z}_n + z_n) S_n - \sum_{k = 1}^{n-1} S_k \left[\left(\tilde{z}_{k+1} - z_{k+1}\right)-\left(\tilde{z}_k - z_k\right)\right],$$
where $S_k := \sum_{i = 1}^k (\tilde{z}_i - z_i)$. Therefore,
$$\abs{\dfrac{1}{n}\sum_{i = 1} (\tilde{z}_i - z_i)^2} \leq 2\left(\sup_{1 \leq i \leq n}\abs{z_i} + \sup_{1 \leq i \leq n}\abs{\tilde{z}_i}\right)\abs{\dfrac{1}{n}\sum_{k = 1}^n S_k}.$$

Suppose now that
$$\sum_{i = 1}^n (\tilde{z}_i - z_i) = o(n^{-a})$$
holds in mean (resp. a.s.). Then, by the results of \citep{bibaut2020sufficient}, it follows that $n^{-1}\sum_{i = 1}^n (\tilde{z}_i - z_i)^2 = o(n^{-a})$ in mean (resp. a.s.) as well.

Otherwise, if $n^{-1}\sum_{i = 1}^n (\tilde{z}_i - z_i) = O(n^{-a})$ in mean or almost surely, then by \cite[Theorem 3.2]{Apostol_1976} (using Euler's summation formula), it follows that:
$$\dfrac{1}{n}\sum_{k = 1}^n S_k = 
\begin{cases}
    O\left(n^{-\max(1, a)}\right); \quad \mathrm{ for } a \neq 1,\\
    O\left(\log(n)/n\right); \quad \mathrm{ for } a = -1
\end{cases},
$$
in mean, or almost surely, and so the same holds for $n^{-1}\sum_{i = 1}^n (\tilde{z}_i - z_i)^2$.
\end{proof}

\begin{proof}[Proposition 4.1]
By the definition of $\widehat{\boldsymbol\beta}$, we have:
$$\dfrac{1}{n}\norm{F(\mathbf{r}_t) - \widehat{\boldsymbol\beta}}^2_2 + \lambda\norm{\mathbf{D}\widehat{\boldsymbol\beta}}_1 \leq \dfrac{1}{n}\norm{F(\mathbf{r}_t)-\boldsymbol{\sigma}^2_t}^2_2 + \lambda\norm{\mathbf{D}\boldsymbol{\sigma}^2_t}_1.$$

We thus have:
$$\dfrac{1}{n}\norm{\widehat{\boldsymbol\beta}-\boldsymbol{\sigma}^2_t}^2_2 + \dfrac{2}{n}\left\langle F(\mathbf{r}_t)-\boldsymbol{\sigma}^2_t,\widehat{\boldsymbol\beta}-\boldsymbol{\sigma}^2_t\right\rangle_2 \leq \lambda\left(\norm{\mathbf{D}\boldsymbol{\sigma}^2_t}_1 - \norm{\mathbf{D}\widehat{\boldsymbol\beta}}_1\right).$$

Using the Cauchy-Schwarz inequality and the fact that $\norm{\cdot}_1 \leq \sqrt{n}\norm{\cdot}_2$, we obtain:
$$\dfrac{1}{n}\norm{\widehat{\boldsymbol\beta}-\boldsymbol{\sigma}^2_t}^2_2 \leq \dfrac{2}{n}\norm{F(\mathbf{r}_t)-\boldsymbol{\sigma}^2_t}_2 \cdot \norm{\widehat{\boldsymbol\beta}-\boldsymbol{\sigma}^2_t}_2 + 2\lambda\sqrt{n}\norm{\widehat{\boldsymbol\beta}-\boldsymbol{\sigma}^2_t}_2.$$

Dividing through by $n^{-1/2}\norm{\widehat{\boldsymbol\beta}-\boldsymbol{\sigma}^2_t}_2$ results in the bound
$$\dfrac{1}{\sqrt{n}}\norm{\widehat{\boldsymbol\beta}-\boldsymbol{\sigma}^2_t}_2 \leq \dfrac{2}{\sqrt{n}}\norm{F(\mathbf{r}_t)-\boldsymbol{\sigma}^2_t}_2 + 2\lambda n.$$

Now using the triangle inequality, we obtain:
$$\dfrac{1}{2\sqrt{n}}\norm{\hat{\boldsymbol{\beta}} - \boldsymbol{\sigma}^2_t}_2 \leq \dfrac{1}{\sqrt{n}}\norm{F(\mathbf{r}_t) - \boldsymbol{\sigma}^2_t - \chi \cdot \boldsymbol{\zeta}^J}_2 + \chi \cdot \dfrac{1}{\sqrt{n}}\norm{\boldsymbol{\zeta}^J}_2 + \lambda n.$$

This establishes the result.
\end{proof}

\end{document}